\documentclass[]{article}
\usepackage[utf8]{inputenc}
\usepackage{enumerate}
\usepackage{cite}
\usepackage{amssymb}
\usepackage{amsthm}
\usepackage[margin=3.0cm]{geometry}
\usepackage{framed}
\usepackage[font={footnotesize}]{caption}
\usepackage{amsmath}
\usepackage{color}
\usepackage{stmaryrd}
\usepackage{tikz}
\usetikzlibrary{arrows}

\makeatletter
\def\@endtheorem{\endtrivlist}
\makeatother

\newtheoremstyle{definition}
{}{}{}{}{\bfseries}{.}{.5em}{{\thmname{#1 }}{\thmnumber{#2}}{\thmnote{ (#3)}}}
\theoremstyle{definition}

\newtheorem{defn}{Definition}

\newtheoremstyle{remboldstyle}
{}{}{\itshape}{}{\bfseries}{.}{.5em}{{\thmname{#1 }}{\thmnumber{#2}}{\thmnote{ (#3)}}}
\theoremstyle{remboldstyle}

\newtheorem{theor}[defn]{Theorem} 
\newtheorem{pro}[defn]{Proposition} 
\newtheorem{lemma}[defn]{Lemma} 
\newtheorem{coro}[defn]{Corollary}

\numberwithin{defn}{section}

\newcommand{\C}{\mathcal{CS}}
\newcommand{\D}{\mathcal{D}}
\newcommand{\W}{\mathcal{W}}
\newcommand{\R}{\mathcal{R}}
\newcommand{\M}{\mathcal{M}}

\newcommand{\E}{\mathcal{E}}
\newcommand{\I}{\mathcal{I}}
\newcommand{\p}{^{\prime}} 
\newcommand{\pp}{^{\prime\prime}}
\newcommand{\bl}{\langle}
\newcommand{\br}{\rangle}
\newcommand{\vvarphi}{\vec{\varphi}}
\newcommand{\tp}{\textbf{p}}
\newcommand{\tq}{\textbf{q}}
\newcommand{\tr}{\textbf{r}}
\newcommand{\tvp}{\vec{\textbf{p}}}
\newcommand{\Fj}{\mathit{Fml}_{J}}
\newcommand{\Tj}{\mathit{Term}_{J}}

\newcommand{\Fjv}{\mathit{Fml}_{J}(\textbf{V})}

\newcommand{\Cv}{\C(\textbf{V})}
\newcommand{\turn}{\vdash}
\newcommand{\turnc}{\vdash_{\C}}
\newcommand{\turncv}{\vdash_{\Cv}}
\newcommand{\nmodels}{\not\Vdash}
\newcommand{\model}{\bl\W,\R,\D,\I,\E \br}

\newcommand{\jtff}{\textsf{JT45}}
\newcommand{\msf}{\textsf{S4}}
\newcommand{\msfi}{\textsf{S5}}
\newcommand{\fomsf}{\textsf{FOS4b}}
\newcommand{\jlp}{\textsf{LP}}
\newcommand{\fojtff}{\textsf{FOJT45}}
\newcommand{\folp}{\textsf{FOLP}}
\newcommand{\fosf}{\textsf{FOS4}}
\newcommand{\folpb}{\textsf{FOLPb}}
\newcommand{\folpbn}{\textsf{FOLPb}_{0}}

\newcommand{\gen}{\mathsf{gen}}
\newcommand{\jb}{\mathsf{b}}
\newcommand{\lc}[1]{#1\!\!:\!\!}

\newcommand{\nec}{\square}

\usepackage[pagebackref]{hyperref}
\usepackage{longtable}

\title{First-order justification logic with constant domain semantics}
\author{Melvin Fitting\thanks{Department of Computer Science (Emeritus), The Graduate School and University Center (CUNY).} \and Felipe Salvatore\thanks{PhD candidate at  the Institute of Mathematics and Statistics. S\~{a}o Paulo University (USP).} \thanks{From $2013$ to $2015$, this author was supported by the S\~{a}o Paulo Research Foundation (FAPESP) Grant $2012/24619-7$. This collaboration was made possible by a visit of this author to New York, thanks again to the financial support of S\~{a}o Paulo Research Foundation Grant $2014/11951-9$.}}

\begin{document}

\maketitle

\begin{abstract}
Justification logic is a term used to identify a relatively new family of modal-like logics. There is an established literature about propositional justification logic, but incursions on the first-order case are scarce. In this paper we present a constant domain semantics for the first-order logic of proofs with the Barcan Formula (\folpb); then we prove Soundness and Completeness Theorems.  A monotonic semantics for a version of this logic without the Barcan Formula is already in the literature, but constant domains requires substantial new machinery, which may prove useful in other contexts as well.  Although we work mainly with one system, we also indicate how to generalize these results for the quantified version of \jtff, the justification counterpart of the modal logic \msfi.  We believe our methods are more generally applicable, but initially examining specific cases should make the work easier to follow.\\

\end{abstract}

\section{Introduction}
\label{intro}
Propositional justification logics are similar to modal logics, except that instead of a modal necessity operator $\nec$ one has an infinite family of \emph{justification terms}.  These are built up from variables and constants using basic operators that depend on the particular justification logic.  In place of $\nec\varphi$ one has $\lc{t}\varphi$ where $t$ is a justification term.  One can think of $\lc{t}\varphi$ epistemically, asserting that `$t$ is a justification for $\varphi$'.  Then not only we can reason about what an agent knows but moreover we can reason about the ground of his/her knowledge: justifications.  The first justification logic \jlp \ (Logic of Proofs) was introduced as part of a solution to a long-standing question about an arithmetic semantics for propositional intuitionistic logic \cite{Artemov01}.  On one side, \jlp\ embeds into formal arithmetic, and on the other side the modal logic \msf\ embeds into \jlp.  For the embedding of \jlp\ into arithmetic, $\lc{t}\varphi$ can be understood as `$t$ is a proof for $\varphi$'; propositional variables can be interpreted as sentences of arithmetic, justification variables and constants can be associated with natural numbers and the basic operators can be related to specific recursive functions. Thus $\lc{t}\varphi$ can be translated in first-order arithmetic as the sentence asserting that the number denoted by $t$ is a proof of the G\"odel number of $\varphi$.  Since the early work, the number of justification logic, modal logic pairings has grown enormously, \cite{Fitting16}.  The connection with formal arithmetic does not extend widely, but the epistemic justification/modal connection remains central.  

In \cite{Artemov01} the following was proved. Given an \jlp\ theorem $\varphi$ if we just forget the justification terms and write boxes instead we will get an \msf\ theorem $\varphi^{\circ}$; and for any \msf\ theorem $\varphi$ there is a procedure $r$ (a \emph{realization}) that enables us to replace different occurrences of boxes with specific justification terms in order to get a  \jlp\ theorem $\varphi^{r}$.  The theorem that establishes this result is known as the \emph{Realization Theorem} for \jlp.  The  notions are defined so that for any modal formula $\varphi$, $(\varphi^r)^\circ = \varphi$. Realization results started with \msf\ but nowadays there are Realization Theorems for an infinite family of modal logics, \cite{Fitting16}.

A quantified version of the justification logic \jlp\ has been created, \cite{Artemov11,Fitting13,Fitting14}, but it has not yet been extensively studied.  In particular, justification versions of possible world models have all had \emph{monotonic} domain functions, and constant domain versions have not been considered.  It is the purpose of this paper to examine a quantified version of \jlp\ with a justification counterpart of the \emph{Barcan Formula}.  An axiomatic system together with a constant domain semantics is presented.  Soundness and Completeness are established.  We also show how to extend the results obtained to the quantified version of justification logic \jtff, the justification counterpart of \msfi.  We believe our work will extend to other quantified justification logics having the Barcan Formula as an axiom, but here we confine ourselves to very concrete instances.  Also, we do not consider corresponding Realization Theorems here, though this will be central to further work.

\section{\folpb: Language and axiom system}

\label{LanguageAxiomSystem}

Above we mentioned two informal interpretations of justification formulas, one epistemic, the other related to formal proofs.  It is the second that provides the best motivation for the machinery of first-order justification logics.  Recall, in the propositional case we use $\lc{t}\varphi$ to, informally, denote that $t$ is a derivation of $\varphi$. When we move to first-order we need to take individual variables into account.  In a first-order derivation, free variables play two roles: a variable can be a \emph{formal symbol} that can be subject to generalization, or it can be a \emph{place-holder} that can be substituted for.  These two roles are not compatible with each other, and distinguishing between them is actually the key to the fact that first-order justification logic can internalize its own proofs, a fundamental and essential property.  We will see formulas of the form $\lc{t}_X\varphi$ where $X$ is a set of individual variables.  Informally we can think of this as expressing that the variables occurring in $X$ have the role of place-holders in the derivation $t$ of $\varphi$, while variables not in $X$ can be universally generalized. This informal motivation will appear only incidentally in this paper, but it is important to be aware of it in order to understand the design of first-order justification logic.  In this new scenario we will deal with formulas such as $\lc{t}_{\{x\}}\varphi(x)$. The role of $\{x\}$ in $\lc{t}_{\{x\}}\varphi(x)$ is to indicate that \emph{$x$ is free throughout the derivation $t$} of $\varphi(x)$, and so if $e$ is an individual term of  the language and we substitute $e$ for $x$ in $t$ we will obtain a derivation $s$ of $\varphi(e)$. 

\begin{defn}[Basic vocabulary]
The symbols used to formulate the language of first-order justification logic have a familiar part common to classical first-order logic, and a part specifically related to justifications. The classical part consists of individual variables $ x_{0}, \, x_{1}, \, x_{2}, \, \dots$; predicate variables of all arities, $ P, \, Q, \, P\p, \, Q\p, \, \dots $; and the logical operations $\rightarrow, \, \bot, \, \forall $.  Other operations are understood as defined, and will be used as needed. As in the \jlp\ propositional case we have justification variables $ p_{0}, \, p_{1}, \, p_{2}, \, \dots $, justification constants $c_{0}, \, c_{1}, \, c_{2}, \, \dots $  and the operators $ +, \, \cdot, \, !$. What is new for the first-order case is that for every individual variable $x$ a new operator $\gen_{x}$ is added; and for every finite set of individual variables $X$ we add the constructor $(\cdot):_{X} (\cdot)$. This was introduced in \cite{Artemov11}. Here we add a new justification term $\jb(\cdot)$. It will play a role in the formulation of the Barcan Formula in justification logic. When working with constant domain first-order \jlp\ $\jb$ will be a primitive symbol; but when we move to first-order \jtff\ this symbol will be replaced with an expression defined from the other formal machinery of justification terms.
\end{defn}

\begin{defn}[Language]
The language of constant domain first-order \jlp\ consists of the set $\Tj$ of justification terms $t$ and of the set $\Fj$ of formulas $\varphi$, formed by the following grammar in which $X$ is a finite set of individual variables:        
\begin{align*} 
t &::= p_{i} \; | \; c \; | \; (t \cdot t) \; | \; (t + t) \; | \; !t \; | \; \jb(t) \; | \; \gen_{x}(t) \; , \\[0.1cm] 
\varphi &::=  Px_1 \dots x_n  \; |  \; \bot  \;|  \; (\varphi \rightarrow \varphi)  \; |  \; \forall x \varphi  \;|  \; \lc{t}_{X}\varphi \; .
\end{align*}
\end{defn}

We are assuming that the set of individual variables, justification variables and justification constants are all countable sets. Then, both $\Tj$ and $\Fj$ are countable sets. 

\begin{defn}[Free Variable Occurrence]\label{freevarocc}
Free variable occurrences are defined recursively as usual, but with one new clause.  The free variable occurrences of $\lc{t}_{X}\psi$, are those of $\psi$, provided the variables also occur in $X$, together with all occurrences in $X$ itself. We will use ${\it fv(\varphi)}$ to denote the set of free variables of $\varphi$. Hence, by the definition of free variable occurrence, ${\it fv(\lc{t}_{X}\varphi) = X}$.
\end{defn}

\begin{defn}[Substitution]
We define the notion of an individual variable $y$ being free for $x$ in the formula $\varphi$. The definition is the same as in the classical case, with the addition of the following clause: $y$ is free for $x$ in $\lc{t}_{X}\varphi$ if two conditions are met, i) $y$ is free for $x$ in $\varphi$, ii) if $y \in fv(\varphi)$, then $y \in X$.

For a formula $\varphi(x_{1}, \dots, x_{n})$ and individual variables $y_{1}, \dots, y_{n}$ such that $y_i$ is free for $x_i$ in $\varphi$, we write $\varphi(y_{1}/x_{1}, \dots, y_{n}/x_{n})$ to denote the formula obtained by substitution of $y_{1}, \dots, y_{n}$ for all the free occurrences of $x_{1}, \dots, x_{n}$ in $\varphi$,  respectively. When it is clear from the context which variables are being substituted for in $\varphi$ we will simply write $\varphi(y_{1}, \dots, y_{n})$.
\end{defn}

We adopt some useful notational conventions concerning finite sets of individual variables.  We write $Xy$ instead of $X \cup \{y\} $, and we assume that $y \notin X$.  We use $\lc{t}\varphi$ as an abbreviation for $\lc{t}_{\emptyset}\varphi$.

\begin{defn} [Basic axiom system]
We begin with an auxiliary axiom system called $\folpbn$ that can be found in Figure~\ref{basicaxioms}. 
\vspace{0.5cm}

\begin{framed}
\begin{tabular}{l l}

\textbf{A1} & Classical axioms of first-order logic\\[0.1cm]

\textbf{A2} & $\lc{t}_{Xy}\varphi \rightarrow\lc{t}_{X}\varphi$, provided $y$ does not occur free in $\varphi$ \\[0.1cm]

\textbf{A3} & $\lc{t}_{X}\varphi \rightarrow \lc{t}_{Xy}\varphi$\\[0.1cm]

\textbf{B1} & $\lc{t}_{X}\varphi \rightarrow \varphi$ \\[0.1cm]

\textbf{B2} & $\lc{t}_{X}(\varphi \rightarrow \psi) \rightarrow (\lc{s}_{X}\varphi \rightarrow \lc{[t\cdot s]}_{X}\psi)$ \\[0.1cm] 

\textbf{B3} & $\lc{t}_{X}\varphi \rightarrow \lc{[t+s]}_{X}\varphi, \, \lc{s}_{X}\varphi \rightarrow \lc{[t+s]}_{X}\varphi$ \\[0.1cm]

\textbf{B4} & $\lc{t}_{X}\varphi \rightarrow \lc{!t}_{X}\lc{t}_{X}\varphi$ \\[0.1cm]

\textbf{B5} & $\lc{t}_{X}\varphi \rightarrow \lc{\gen_{x}(t)}_{X}\forall x \varphi$, provided $x \notin X$ \\[0.1cm]                

\textbf{Bb} & $\forall y \lc{t}_{Xy}\varphi(y) \rightarrow \lc{\jb(t)}_{X} \forall y \varphi(y)$ \\[0.1cm]

\textbf{R1} (\emph{Modus ponens}) & $\turn \varphi, \, \turn \varphi\rightarrow\psi \, \Rightarrow \; \turn \psi$ \\[0.1cm]

\textbf{R2} (\emph{Generalization}) & $\turn \varphi \Rightarrow \; \turn \forall x \varphi$\\[0.1cm]

\end{tabular}
\end{framed}
\captionof{figure}{Axioms schemes and inference rules of $\folpbn$.\protect\label{basicaxioms}}
\vspace{0.5cm}  
\end{defn}

There is no necessitation rule given in Figure~\ref{basicaxioms}.  A necessitation rule can be derived using the machinery of constant specifications, which we discuss next.  The idea is, axioms are not further analyzed and are simply (and arbitrarily) assigned constants that represent justifications for them.  Details may depend, for instance, on what axioms for classical logic are chosen for \textbf{A1} in Figure~\ref{basicaxioms}.

\begin{defn} [Constant specification]   
A set $\C \subseteq \Fj$ is called a \emph{constant specification} if every member of $\C$ is of the form $ \lc{c}\psi$, where $\psi$ is an axiom and $c$ is a constant. We say that $\C$ is \emph{axiomatically appropriate}, if for every axiom $\psi$ there is a constant $c$ such that $\lc{c}\psi \in \C$.
\end{defn}  

\begin{defn} [\folpb] 
The subscript $0$ in the name $\folpbn$ is intended to tell us this axiom system has an empty constant specification.  We write $\folpb_{\C}$ to denote the axiom system that extends $\folpbn$ by adding the formulas from $\C$ as new axioms. By \emph{constant domain first-order \jlp}  \ (\folpb) we mean the family of axiom systems $\folpb_{\C}$ where $\C$ is a constant specification. The notion of derivation, $\Gamma \turn \varphi$, is defined as usual --- it must be noted that, if $\Gamma$ deduces $\varphi$ using the generalization rule, then this rule was not applied to a variable that occurs free in the formulas of $\Gamma$. We use $\Gamma \turn_{\C} \varphi$ to indicate that $\varphi$ is provable in $\folpb_{\C}$ by assuming $\Gamma$.
\end{defn}  

The axiom \textbf{Bb} is a justification analog of the Barcan formula. This axiom is indeed independent of the other ones, by the following argument. By \cite{Artemov11} \folp\ is a counterpart of \fosf, thus if every instance of the axiom \textbf{Bb} were provable in \folp, then every instance of the Barcan formula would be provable in \fosf; and this is not the case.     

\begin{lemma}[Deduction]\label{deductionlemma}
For every constant specification $\C$,
\[
\Gamma,\varphi \turn_\C \psi\; \mbox{ iff } \; \Gamma \turn_\C \varphi \rightarrow \psi \; .
\]
\end{lemma}

And now, the result that takes the place of a primitive necessitation rule.

\begin{theor} [Internalization]\label{internalizationtheorem}
Let $\C$ be an axiomatically appropriate constant specification; $p_{0}, \, \dots, \, p_{k}$ be justification variables; $X_{0}, \, \dots, \, X_{k}$ be finite sets of individual variables, and $X =X_{0} \cup \,\dots \, \cup X_{k}$. If  $\lc{p_0}_{X_0}\varphi_0, \, \ldots, \, \lc{p_k}_{X_k}\varphi_k \turn_\C \psi$, then there is a justification term $t(p_{0}, \, \dots, \, p_{k}) \in \Tj$ such that $\lc{p_0}_{X_0}\varphi_0, \, \ldots, \, \lc{p_k}_{X_k}\varphi_k \turn_\C \lc{t}_X\psi$.
\end{theor}

\begin{proof}
The same proof as presented in \cite[p.\ 7]{Artemov11}.
\end{proof}

\begin{theor}[Converse Barcan and Buridan]
\label{cbarcanburidan}
Let $\C$ be an axiomatically appropriate constant specification and $y$ be an individual variable. For every finite set of individual variables $X$ such that $y \notin X$, every formula $\varphi(y)$, and every justification term $t$, there are justification terms $f(t)$ and $s(t)$ such that:
\begin{enumerate} 
\item $ \turn_{\C} \lc{t}_{X} \forall y \varphi(y) \rightarrow \forall y \lc{f(t)}_{Xy}\varphi(y) \; .$
\item $\turn_{\C} \exists y \lc{t}_{Xy} \varphi(y) \rightarrow \lc{s(t)}_{X} \exists y \varphi(y) \; .$
\end{enumerate}
Item 1 is the explicit counterpart of the converse Barcan Formula; item 2 is the explicit counterpart of the converse Buridan Formula.    
\end{theor}

\begin{proof}
First, the derivation of the explicit converse Barcan formula.
\begin{center}  
\begin{tabular}{l l}
1.  $ \forall y \varphi(y) \rightarrow \varphi(y)$ &  \emph{classical axiom.} \\[0.3cm]
2.  $\lc{c_{1}}(\forall y \varphi(y) \rightarrow \varphi(y))$ &  \emph{constant specification.} \\[0.3cm]       
3.   $\lc{c_{1}}_{Xy}(\forall y \varphi(y) \rightarrow \varphi(y))$ &  \emph{ from 2 by \textbf{A3}.} \\[0.3cm]
4.  $\lc{c_{1}}_{Xy}(\forall y \varphi(y) \rightarrow \varphi(y)) \rightarrow  (\lc{t}_{Xy}\forall y \varphi(y)  \rightarrow   \lc{[c_{1} \cdot t]}_{Xy} \varphi(y))$ &  \emph{\textbf{B2}.} \\[0.3cm]
5.  $\lc{t}_{Xy}\forall y \varphi(y)  \rightarrow   \lc{[c_{1} \cdot t]}_{Xy} \varphi(y)$ &   \emph{from 4 using modus ponens.}  \\[0.3cm]
6.  $\lc{t}_{X}\forall y \varphi(y)  \rightarrow   \lc{t}_{Xy}\forall y \varphi(y)$ &  \emph{ \textbf{A3}.} \\[0.3cm]       
7.  $\lc{t}_{X}\forall y \varphi(y)  \rightarrow    \lc{[c_{1} \cdot t]}_{Xy} \varphi(y)$ &  \emph{from 5 and 6.} \\[0.3cm]
8.  $\forall y(\lc{t}_{X}\forall y \varphi(y)  \rightarrow    \lc{[c_{1} \cdot t]}_{Xy} \varphi(y))$ &  \emph{generalization.} \\[0.3cm]
9.  $\lc{t}_{X}\forall y \varphi(y)  \rightarrow    \forall y\lc{[c_{1} \cdot t]}_{Xy} \varphi(y)$ &  \emph{$y \notin X$ and classical reasoning.} \\[0.3cm]        
\end{tabular}
\end{center}
Next the explicit converse Buridan Formula.

\begin{center}  
\begin{tabular}{l l}
1.  $ \forall y (\varphi(y)  \rightarrow \exists y \varphi(y))$ &  \emph{classical validity.} \\[0.3cm]
2.  $\lc{r}\forall y (\varphi(y) \rightarrow \exists y \varphi(y))$   & \emph{using internalization.} \\[0.3cm]
3. $\lc{r}_{X}\forall  y (\varphi(y) \rightarrow \exists y \varphi(y))$ & \emph{from 2 by \textbf{A3}.}  \\[0.3cm]
4. $\forall  y \lc{f(r)}_{Xy} (\varphi(y) \rightarrow \exists y \varphi(y))$  & \emph{from 3 using conv.\ Barcan.} \\[0.3cm]
5. $\lc{f(r)}_{Xy}(\varphi(y) \rightarrow \exists y \varphi(y))$   & \emph{from 4 and classical axioms.}\\[0.3cm]   
6. $\lc{t}_{Xy}\varphi(y) \rightarrow \lc{[f(r) \cdot t]}_{Xy}\exists y \varphi(y)$   &  \emph{from 5 using \textbf{B2}.}  \\[0.3cm]
7. $\lc{[f(r) \cdot t]}_{Xy}\exists y \varphi(y) \rightarrow \lc{[f(r) \cdot t]}_{X}\exists y \varphi(y)$ & \emph{\textbf{A2}.}  \\[0.3cm]
8. $\lc{t}_{Xy}\varphi(y) \rightarrow \lc{[f(r) \cdot t]}_{X}\exists y \varphi(y)$     &  \emph{from 6 and 7.}\\[0.3cm]
9. $\forall  y (\lc{t}_{Xy}\varphi(y) \rightarrow \lc{[f(r) \cdot t]}_{X}\exists y \varphi(y))$  &  \emph{generalization.} \\[0.3cm]
10. $\exists y \lc{t}_{Xy}\varphi(y) \rightarrow \lc{[f(r) \cdot t]}_{X}\exists y \varphi(y)$  & \emph{$y \notin X$ and classical reasoning.} \\[0.3cm]
\end{tabular}
\end{center}
\end{proof} 

As the reader can see, the terms $f$ and $s$ constructed above depend not only on $t$ but also on the formula $\varphi$ and the constant specification $\C$ being considered. For simplicity's sake, we decided to suppress some details in the notation.  Also it should be noted that axiom \textbf{Bb} was not used in these proofs. 

\section{Semantics}\label{semantics}

In propositional modal logic one uses the notion of \emph{frame}. In the quantified case we have the analogous notion of \emph{skeleton}. Skeletons are also used in the semantics of  first-order justification logic.

\begin{defn}[Skeleton]
\label{defskeleton}
A \emph{skeleton} is a triple $\bl \W, \R, \D \br$, where $\W\neq \emptyset$, $\R \subseteq \W \times \W$ and $\D$ is a non-empty set, the \emph{domain} of the skeleton.\footnote{Since we will be working with constant domain models only we do not need to consider different domains associated to each world $w \in \W$, so $\D$ will be always a set and not a function.} Given a skeleton $\bl \W, \R, \D \br$, we call it an \emph{\folpb-skeleton} if $\R$ is a reflexive and transitive relation on $\W$. Similarly, we call it an \emph{\fojtff-skeleton} if $\R$ is an equivalence relation.   
\end{defn}  

\noindent\textbf{Conventions and Notation} \label{conventionandnotation:} For the rest of the paper we will use $\vec{x}, \vec{y}, \vec{z}, \dots $ as  finite sequences or vectors of individual variables where the length can be inferred from context. Likewise we write $\vec{a}$, $\vec{b}$, $\vec{c}$, \ldots for vectors of members of the domain $\D$ of a model ($c$ is also used to denote a justification constant; the two distinct uses of this symbol will  be clear from context).  For the sake of simplicity, if $\vec{a}$ is a $n$-uple of elements of $\D$ we will write `$\vec{a} \in \D$' instead of $\vec{a} \in \D^{n}$. In the same manner, we will write `$\{\vec{a}\}$' instead of `$\{a_{1}, \dots , a_{n}\}$'. As before, for a formula $\varphi(\vec{x})$ and domain members $\vec{a}$, we use $\varphi (\vec{a})$ as a shorthand for the obvious substitution.  For term $t$ and $\vec{x} = \langle x_{1}, x_{2}, \dots, x_{n}\rangle$ we will write  `$\gen_{\vec{x}}(t)$' to abbreviate the term `$\gen_{x_{1}}(\gen_{x_{2}} \dots (\gen_{x_{n}}(t))\ldots)$'.

Ordinarily in dealing with models one uses \emph{valuations}, mapping variables to members of some domain $\D$.  Thus one talks about the behavior of a formula $\varphi$ with respect to valuation $v$.  In this paper things are already somewhat complicated, so we follow a common procedure of allowing members of the domain $\D$ to appear directly in formulas as if they were variables---we call such things \emph{$\D$ formulas}.  Technically the result is not a formula, but any assertion about it can easily be rephrased using proper formulas and valuations.  

\begin{defn} [$\D$-formulas]
\label{DFormulas}
Let $\D$ be a non-empty set.  If $\vec{x} = \langle x_1, \ldots, x_n\rangle$ is a vector of variables, $\vec{d} = \langle d_1, \ldots, d_n\rangle$ is a vector of members of domain $\D$, and $\varphi(\vec{x})$ is a formula, we write $\varphi(\vec{d})$ for what we call a \emph{$\D$ formula} that results by having $d_i$ in place of each free occurrence of $x_i$, $i=1, \ldots, n$, Definition~\ref{freevarocc}.  We use $\D$-$\Fj$ for the set of all \emph{$\D$-formulas}.  Free individual variables can occur in $\D$ formulas. For a $\D$-formula $\varphi$, we say that $\varphi$ is closed if $\varphi$ has no individual free variable occurrences, though it may contain members of $\D$.
\end{defn}

We emphasize that the use of $\D$ formulas can easily be replaced with valuations.  It is simply a notational convenience.  Thus for $\vec{x} = \langle x_1, \ldots, x_n\rangle$ a vector of variables, $\vec{d} = \langle d_1, \ldots, d_n\rangle$ a vector of members of domain $\D$, and $\varphi(\vec{x})$ a formula, instead of speaking about the $\D$ formula $\varphi(\vec{d})$ we could speak about $\varphi(\vec{x})$ with respect to the valuation $v$ that is such that $v(x_i) = d_i$.  These conventions were also used in \cite{Fitting14}.  Note that while domain members in $\D$ formulas have some of the aspects of individual variables, they are never quantified and do not appear as subscripts of $\gen$.

Here is an illustrative example (consulting Definition~\ref{freevarocc} may be useful).  Suppose the set $\D$ contains $a$, $b$, $c$, and $d$, and possibly other things.  Suppose valuation $v$ is such that $v(x) = a$, $v(y)=b$, $v(z)=c$, and $v(w) = d$.  And suppose $P$ is a three-place predicate symbol.  Then $P(x,z,w)$, considered with respect to $v$, would be written simply as $P(a, c,d)$, and $\lc{t}_{\{x,y\}}P(x,z,w)$, considered with respect to $v$, would be written $\lc{t}_{\{a,b\}}P(a,z,w)$, since $z$ and $w$ are not free in $\lc{t}_{\{x,y\}}P(x,z,w)$ while $x$ and $y$ are.  

In Section~\ref{LanguageAxiomSystem} we discussed the informal meaning of variables occurring in a subscript set: $\lc{t}_{\{x\}}\varphi(x)$ says that $x$ is free throughout the derivation $t$ of $\varphi(x)$.  This understanding carries over to $\D$ formulas quite naturally.  For a domain constant $a$, $t(a)$ informally is a derivation of $\varphi(a)$ and, since $a$ is a domain constant, it is not subjected to generalization.  The effect is that it must behave like a free occurrence throughout the derivation $t(a)$ of $\varphi(a)$.

\begin{defn} [Fitting models]
A \emph{Fitting model} is a structure $\M = \model$ with the following constituents and terminology.  $\bl \W, \R, \D \br$ is a skeleton; we refer to $\D$ as the domain of the model, and similarly for $\W$ and $\R$.  $\I$ is an \emph{interpretation function}, assigning to each $n$-ary relational symbol $P$ and each possible world $w$ an $n$-ary relation $\I(P,w)$ on $\D$.  $\E$ is an \emph{evidence function}; for any justification term $t$ and \emph{$\D$-formula} $\varphi$, $\E(t,\varphi) \subseteq \W$.
\end{defn}

The standard informal understanding of evidence functions is: $\E(t,\varphi)$ is the set of possible worlds at which $t$ can be considered to be relevant (not necessarily correct) evidence for $\varphi$.

\begin{defn}[Evidence function conditions] 
\label{evidencefunctinconditions}
Let $\M = \model$ be a Fitting model. We require the evidence function to meet the following conditions:
\begin{itemize} 
\item[] \textbf{$\cdot$ Condition} $\E(t, \varphi \rightarrow \psi) \cap \E(s, \varphi) \subseteq \E([t\cdot s], \psi).$
\item[] \textbf{$+$ Condition} $\E (s, \varphi) \cup \E(t, \varphi) \subseteq \E([s+t], \varphi).$
\item[] \textbf{$!$ Condition} $\E (t, \varphi) \subseteq \E(!t, \lc{t}_{X} \varphi)$, where $X$ is a set of members of $\D$ including those occurring in $\varphi$.
\item[] \textbf{$\R$ Closure Condition} If $w \in \E (t, \varphi)$ and $w \R w\p$, then $w\p \in \E (t, \varphi)$.
\item[] \textbf{Instantiation Condition} If $w \in \E (t, \varphi(x))$ and $a \in \D$, then $w \in \E (t, \varphi(a))$.
\item[] \textbf{$\gen_{x}$ Condition} $\E (t, \varphi) \subseteq \E(\gen_{x}(t),\forall x\varphi)$.
\item[] \textbf{$\jb$ Condition} If for every $a \in \D$ we have that $w \in \E (t, \varphi(a))$, then $w \in \E (\jb(t), \forall y \varphi(y))$.                   
\end{itemize}

We say that a model $\M = \model$ \emph{meets constant specification $\C$} iff whenever $\lc{c}\varphi \in \C$, then $\E (c, \varphi) = \W$.
\end{defn}

The evidence function conditions reflect the axioms of \folpb\ into the structure of the skeleton. The instantiation condition codifies the axiom \textbf{A3}. And as the name suggests, the $\jb$ Condition reflects the \textbf{Bb} axiom. Although the $\jb$ and the $\gen_{x}$ conditions have a somewhat similar form, they have a very different meaning. Recalling the provability reading of justification statements, we can read the axiom \textbf{B5} as follows: if $t$ is a derivation of $\varphi(x)$ in which $x$ is not free, then $\gen_{x}(t)$ is a derivation of $\forall x \varphi$; informally we extend the derivation $t$ of $\varphi(x)$ with one application of the generalization rule to get a derivation $\gen_{x}(t)$ of $\forall x\varphi$. But axiom \textbf{Bb} express a different idea.  Fix a domain $\D$. Let $t(x)$ be a derivation and $\varphi(x)$ a formula. If for every $a \in \D$, $t(a)$ is a derivation of $\varphi(a)$, then $\jb(t)$ is a derivation of $\forall x \varphi$. This axiom embodies a kind of infinitary inference rule resembling the $\omega$-rule. In a first-order language expanded with a new set of constants $\{ a_{\xi} \;|\; \xi < \kappa\}$ (where $\kappa = |\D|$) this rule can be formulated as follows: $\vdash \varphi (a_{\xi}) \; \emph{for every }\xi < \kappa \; \Rightarrow \;\vdash \forall x \varphi (x)$.  On the one hand the justification term $\gen_{x}$ captures the introduction of the universal quantifier via the generalization rule. On the other hand, the new justification term $\jb$ internalizes the introduction of the universal quantification via an infinitary inference rule, an $\omega$-rule.

\begin{defn}[Truth at words]
Let $\M = \model$ be a Fitting model, $\varphi$ a closed $\D$-formula and $w \in \W$. The notion that \emph{$\varphi$ is true at world $w$ of $\M$}, in symbols $\M,w \Vdash  \varphi$, is defined recursively as follows: 

\begin{itemize} 
\item $\M,w \Vdash  P(\vec{a})$ iff $\vec{a} \in \I(P,w)$. 
\item $\M,w \nmodels \bot$. 
\item $\M,w \Vdash  \psi \rightarrow \theta$ iff $\M,w \nmodels \psi$ or $\M,w \Vdash  \theta$.
\item $\M,w \Vdash  \forall x \psi(x)$ iff for every $a \in \D$, $\M,w \Vdash  \psi(a)$.        
\item Assume $\lc{t}_{X}\psi(\vec{x})$ is closed and $\vec{x}$ are all the free variables of $\psi$. Then, $\M,w \Vdash  \lc{t}_{X}\psi(\vec{x})$ iff
\begin{enumerate}[(a)]
\item $w \in \E (t, \psi(\vec{x}))$ and
\item for every $w\p \in \W$ such that $w\R w\p$, $\M,w\p \Vdash  \psi(\vec{a})$ for every $\vec{a} \in \D$.
\end{enumerate}
\end{itemize}
\end{defn}

The definition above covers only closed formulas that can contain members of $\D$.   The following is what we finally are interested in, validity for formulas that can contain free variables but no members of $\D$.

\begin{defn} [Validity] 
Let $\varphi \in \Fj$ be a closed formula. We say that $\varphi$ is \emph{valid in the Fitting model} $\M = \model$ provided for every $w \in \W$, $\M,w \Vdash  \varphi$. A formula with free individual variables is valid if its universal closure is valid.
\end{defn}

\begin{defn}[Fitting model for \folpb]
A \emph{Fitting model for \folpb} is a Fitting model $\M = \model$ where $\bl \W, \R, \D \br$ is an \folpb-skeleton.
\end{defn}

For a formula $\varphi$ and constant specification $\C$, we write $\Vdash _{\C}\varphi$ if $\varphi$ is valid in $\M$ for every Fitting model for \folpb\ $\M$ meeting $\C$.

\section{Soundness}

We prove soundness with respect to an arbitrary constant specification.  Completeness, shown in Section~\ref{canonicalmodelsection}, will be a narrower result, requiring axiomatic appropriateness.

\begin{theor}[Soundness]
Let $\C$ be a constant specification. For every formula $\varphi \in \Fj$, if $\turn_{\C} \varphi$, then $\Vdash _{\C}\varphi$.
\end{theor}    

\begin{proof}
The proof is by induction on proof length. The argument is almost the same as presented in \cite{Fitting14}. In that paper, the case for the \textbf{B4} axiom is not properly proved.  We will show validity for the \textbf{B4} axiom, and also for the axiom \textbf{Bb}.
\smallskip

Suppose $\varphi$ is a representative special case of \textbf{B4}; take $\varphi$ to be $\lc{t}_{\{x,z\}}\psi(x,y) \rightarrow \lc{!t}_{\{x,z\}}\lc{t}_{\{x,z\}}\psi(x,y)$. Let $\M = \model$ be a Fitting model for \folpb\ meeting $\C$ with $w \in \W$ arbitrary.  We show the universal closure of $\varphi$ is true at $w$ and, since $x$ and $y$ have free occurrences but $z$ does not, it is enough to show $\lc{t}_{\{a,b\}}\psi(a,y) \rightarrow \lc{!t}_{\{a,b\}}\lc{t}_{\{a,b\}}\psi(a,y)$ is true at $w$ for arbitrary $a,b\in\D$.
\smallskip

Assume $\M, w \Vdash \lc{t}_{\{a,b\}}\psi(a,y)$. Then we must have $w \in \E(t, \psi(a,y))$, so by the $!$ Condition, $w \in \E(!t, \lc{t}_{\{a,b\}}\psi(a,y))$.  Now let $v$ and $u$ be arbitrary members of $\D$ with $w\R v$ and $v\R u$. By the transitivity of $\R$, $w\R u$. Since $\M, w \Vdash \lc{t}_{\{a,b\}}\psi(a,y)$, we have that $\M, u \Vdash   \psi(a,d)$ for every $d \in \D$. Also $w \in \E(t, \psi(a,y))$ so by the $\R$ Closure Condition, $v \in \E(t, \psi(a,y))$. Then, since $u$ is arbitrary, $\M, v \Vdash \lc{t}_{\{a,b\}}\psi(a,y)$. And since the choice of $v$ was arbitrary, $\M, w \Vdash \lc{!t}_{\{a,b\}}\lc{t}_{\{a,b\}}\psi(a,y)$.     
\smallskip    

Next, suppose $\varphi$ is an instance of \textbf{Bb}, say $\varphi$ is $\forall y \lc{t}_{Xy}\psi(y) \rightarrow \lc{\jb(t)}_{X} \forall y \psi(y)$, and for simplicity assume $X= \{x\}$ and $\psi = \psi(x,y)$, so our particular instance is $\forall y \lc{t}_{\{x,y\}}\psi(x,y) \rightarrow \lc{\jb(t)}_{\{x\}} \forall y \psi(x,y)$.  We show validity of its universal closure.  Let $\M = \model$ be a Fitting model for \folpb\ meeting $\C$, let $w \in \W$ be arbitrary, and let $a\in \D$ be arbitrary.  Assume $\M, w \Vdash  \forall y \lc{t}_{\{a,y\}}\psi(a,y)$; we show $\M, w \Vdash \lc{\jb(t)}_{\{a\}} \forall y \psi(a,y)$

By our assumption, $\M, w \Vdash   \lc{t}_{\{a,b\}}\psi(a,b)$ for every $b \in \D$. Then $w \in \E(t,\psi(a,b))$ for every $b\in \D$, so by the $\jb$ Condition, $w \in \E(\jb(t),\forall y\psi(a,y))$.  Now let $v \in \W$ be arbitrary such that $w\R v$. Since for every $b\in \D$ we have $\M, w \Vdash   \lc{t}_{\{a,b\}}\psi(a,b)$, then $\M, v \Vdash \psi(a,b)$ for every $b\in\D$, and hence  $\M, v \Vdash \forall y \psi(a,y)$. Since $v$ was arbitrary, $\M, w \Vdash \lc{\jb(t)}_{\{a\}} \forall y \psi(a,y)$.   
\end{proof}

\section{Language Extension}\label{langextsection}

In order to prove a Completeness Theorem, in Sections~\ref{henkinsection} and \ref{canonicalmodelsection} we will construct a canonical model using a Henkin style argument. As in Section~\ref{semantics} we allow additional things to appear in formulas, but now they come from a new countable set of variables $\textbf{V}$ that we call `witness variables'. These witness variables can appear in formulas, but are never quantified and do not appear as subscripts on $\gen$.  This way of extending the language has some differences from the one presented in \cite{Fitting14}.  Although here we are using a new type of variable, it essentially has a role similar to a domain constant.  It is made a variable partly to help with readability, since we also work with justification constants.

From now on we use the phrase \emph{individual variables} to denote the members of $\textbf{V}\cup \{x_{0}, \, x_{1}, \, \dots \}$, \emph{basic variables} to denote the members of $\{x_{0}, \, x_{1}, \, \dots \}$, our original variables, and \emph{witness variables} to denote the newly added members of $\textbf{V}$. In typical Henkin style we are interested in taking $\textbf{V}$ to be the domain $\D$ of the canonical model we will construct, so for the course of the completeness proof we will call a \emph{Henkin formula} a first-order justification formula in which basic variables may occur free or bound, but variables from $\textbf{V}$ may only occur free (a Henkin formula is just a $\D$-formula when $\D = \textbf{V}$).  Symbolically we write $\Fjv$ for this language.  The set of justification terms is not enarged since witness variables do not appear as subscripts on $\gen$, so both $\Fj$ and $\Fjv$ have $\Tj$ as their set of justification terms.  We call a Henkin formula closed if no \emph{basic} variable occurrences are free. Using this enlarged language we construct a new axiomatic system for \folpb\ based on the formulas from $\Fjv$. By \emph{basic system} we mean the language and the axiomatic system presented earlier, without the extension to Henkin formulas.

\begin{defn} [Variable variants]
Two Henkin formulas are \emph{variable variants} provided each can be turned into the other by a  renaming of free individual variables. 
In other words, there is a 1-1, onto map from the free variables of one formula to the feee variables of the other.
\end{defn}

We will be interested in variable variants within the basic language, and also in the language with witness variables added.  In either case, this does not concern quantified variables or those that appear as $\gen$ subscripts.

\begin{defn}[Variant closed] A constant specification $\C$ is \emph{variant closed} iff whenever $\varphi$ and $\psi$ are variable variants, then $\lc{c}\varphi \in \C$ iff $\lc{c}\psi \in \C$.
\end{defn}

\begin{defn}\label{csextensiondef}
Let $\C$ be a variant closed constant specification \emph{for the basic system}. By $\Cv$ we mean the smallest set satisfying the condition: If $\lc{c}\varphi \in \C$, $\psi \in \Fjv$ and $\psi$ is the result of replacing some free basic variables in $\varphi$ with distinct witness variables, then $\lc{c}\psi \in \Cv$ (In other words, $\psi$ is a variable variant of $\varphi$ but in the language allowing witness variables).
\end{defn}

There are several items concerning this that are easy to establish, but are fundamental.

\begin{pro}
\label{basicstuff}
Assume $\C$ is a variant closed constant specification \emph{for the basic system}.
\begin{enumerate}
\item $\C \subseteq \Cv$.
\item $\Cv$ is variant closed.
\item If $\C$ is axiomatically appropriate with respect to $\Fj$, then $\Cv$ is axiomatically appropriate with respect to $\Fjv$. 
\item If $\lc{c}\psi\in\Cv$  and $\psi$ contains no witness variables, then $\lc{c}\psi\in\C$.
\item The Deduction Lemma~\ref{deductionlemma}, the Internalization Theorem~\ref{internalizationtheorem}, and Theorem \ref{cbarcanburidan} hold for the enlarged axiom system allowing witness variables.
\end{enumerate}
\end{pro}

\begin{proof} We take them in order.
\begin{enumerate}
\item In constructing $\Cv$, identity replacement of basic variables is allowed.
\item Suppose $\lc{c}\psi_1\in\Cv$ and $\psi_1$ and $\psi_2$ are variable variants, allowing witness variables.  Since $\lc{c}\psi_1\in\Cv$, there is some $\varphi$ in the basic language that is a variable variant of $\psi_1$ with $\lc{c}\varphi\in\C$.  Since composition of 1-1, onto maps is 1-1 and onto, $\psi_2$ and $\varphi$ are variable variants, and hence $\lc{c}\psi_2\in\Cv$.  

\item Our axiomatization is by \emph{schemes}.
\item Assume $\lc{c}\psi\in\Cv$, where $\psi$ is in the basic language.  There must be some $\lc{c}\varphi\in\C$ where $\varphi$ and $\psi$ are variable variants.  Then $\lc{c}\psi\in\C$ since $\C$ is variant closed.
\item We omit the arguments.
\end{enumerate}
\end{proof}

The next Proposition has some corollaries that are important in our completeness proof.

\begin{pro}
\label{preconservativ}
Assume the following.
\begin{enumerate}
\item $\C$ is a variant closed constant specification for the basic system and $\Cv$ is its extension to $\Fjv$.
\item $\psi_{1}, \, \psi_{2}, \, \dots , \, \psi_{n}$ is an \folpb\ proof in the language of $\Fjv$ using $\Cv$.
\item $a$ is a witness variable that occurs in the proof.
\item $y$ is a basic variable that does not occur in the proof (free, bound, or as a subscript of $\gen$).
\item $(\psi_{i})^{-}$ is the result of replacing $a$ with $y$.
\end{enumerate}
Then $(\psi_{1})^{-}, \, (\psi_{2})^{-}, \, \dots, \, (\psi_{n})^{-}$ is also a proof.
\end{pro}

\begin{proof}
The argument is by induction on proof step.

If $\psi_{i}$ is an axiom, since we are using axiom schemes and $y$ is new, then $(\psi_{i})^{-}$ is also an axiom.

Suppose $\psi_{i}$ is a member of $\Cv$.  Then it is of the form $\lc{c}\phi$, and $(\psi_i)^-$ is $\lc{c}(\phi)^{-}$ where $(\phi)^{-}$ is the result of replacing $a$ by $y$ in $\phi$.  Since $\phi$ and $(\phi)^{-}$ will be variable variants, $\lc{c}(\phi)^{-}\in\Cv$ since it is variant closed.

If $\psi_{i}$ is deduced from $\psi_{i_1}$ and $\psi_{i_2} = \psi_{i_1}\rightarrow \psi_{i}$ by modus ponens, then $(\psi_{i_2})^{-}$ is $(\psi_{i_1})^{-} \rightarrow (\psi_{i})^{-}$. So $(\psi_{i})^{-}$ also follows from $(\psi_{i_2})^{-}$ and $(\psi_{i_1})^{-}$ by modus ponens. 

If $\psi_{i}$ is deduced from $\psi_{l}$ by generalization, then $\psi_{i}$ is $\forall x \psi_{l}$, where $x$ muxt be a basic variable since witness variables are not quantified.  Then $\forall x (\psi_{l})^{-}$ is also deduced from $(\psi_{l})^{-}$ by generalization, since $y$ is new and hence does not affect $x$.   
\end{proof}

\begin{coro}[Conservativity]
\label{conservativ} 
Let $\C$ be a variant closed constant specification for the basic system and $\Cv$ be its extension to $\Fjv$.  For every $\varphi \in \Fj$, if $\turncv \varphi$, then $\turnc \varphi$.
\end{coro}

\begin{proof}
If $\turncv \varphi$, then $\varphi$ has a proof that may contain witness variables.  One by one, replace these with new basic variables.  Using Proposition~\ref{preconservativ}, the result must be a proof, and will contain no witness variables.  By item~4 of Proposition~\ref{basicstuff}, this will be a proof using constant specification $\C$, and since $\varphi$ contained no witness variables, their replacement does not affect it.  Hence there is a proof of $\varphi$ using $\C$. 
\end{proof}

\begin{coro}[Generalization]
\label{gencorollary}
Suppose $\C$ is a variant closed constant specification for the basic language and $\Cv$ is its extension.  Let $\varphi(a)$ be a formula containing witness variable $a$ (and possibly other witness variables).  If $\turncv \varphi(a)$ then for some new basic variable $y$, $\turncv \forall y\varphi(y)$.
\end{coro}

\begin{proof}
Assume $\turncv \varphi(a)$, so it has a proof in the language $\Fjv$.  Let $y$ be any basic variable that does not occur in the proof, and replace occurrences of $a$ throughout that proof with occurrences of $y$.  By Proposition~\ref{preconservativ}, this will still be a proof, but of $\varphi(y)$.  Now Generalization can be applied, Rule \textbf{R2} from Figure~\ref{basicaxioms}, and $\forall y\varphi(y)$ is provable.
\end{proof}

\begin{defn}
Let $\C$ be a variant closed constant specification for the basic language and let $\Gamma \subseteq \Fj$. We say that $\Gamma $ is $\C$-\emph{inconsistent} iff $\Gamma \turn_{\C} \bot$. By the Deduction Lemma~\ref{deductionlemma}, $\Gamma$ is $\C$-inconsistent iff there is a finite subset $\{\psi_{1}, \, \dots, \,  \psi_{n}\}$ of $\Gamma$ such that $\turn_{\C} (\psi_{1} \wedge \,\dots \,\wedge \psi_{n}) \rightarrow \bot$. A set $\Gamma$ is $\C$-\emph{consistent} if it is not $\C$-inconsistent. And we say that $\Gamma$ is $\C$-\emph{maximally consistent} whenever $\Gamma$ is $\C$-consistent and $\Gamma$ has no proper extension that is $\C$-consistent. There are similar notions for $\C(\textbf{V})$.
\end{defn}

It follows from Corollary~\ref{conservativ} that for every set of basic formulas $\Gamma$, if $\Gamma$ is $\C$-consistent, then $\Gamma$ is $\C(\textbf{V})$-consistent.

\begin{pro}[Lindenbaum]
\label{lind}
Let $\C$ be a variant closed constant specification and $\Cv$ its extension. If $\Gamma \subseteq \Fjv$ is $\Cv$-consistent then there is a $\Gamma\p \subseteq \Fjv$ such that $\Gamma \subseteq \Gamma\p$ and $\Gamma\p$ is a $\Cv$-maximally consistent set.
\end{pro}

\section{Templates}
Starting with Section~\ref{henkinsection} we will prove completeness for \folpb\ with respect to a justification logic analog of constant domain Kripke models for \msf. For this we follow the general outline of the constant domain completeness proof in \cite{Hughes96}, but complexities arise in transferring the proof from a modal setting to a justification logic one.  The modal completeness argument makes much use of \emph{finite} sets of modal formulas.  This is important because a finite set of formulas acts like its conjunction --- a single formula.  But if modal operators in a finite set of formulas are replaced with justification terms, an infinite set can result.  There is only one necessity operator, but there are infinitely many justification terms that can replace it.  Somehow finiteness must be restored for a completeness argument to go through.  The key idea is that justification formulas having the same underlying propositional modal structure can be grouped together.  We introduce the notion of \emph{template} as a way of making precise what we need about the modal structure of justification formulas.  Formal details follow after some motivating informal remarks.

Consider the modal formula $\Box(\tp \wedge \Box \psi)$, where $\tp$ is a propositional letter. In any normal modal logic, if $\varphi \rightarrow \psi$ is provable then $\Box(\tp \wedge \Box \varphi) \rightarrow \Box(\tp \wedge \Box \psi)$ will also be provable. The only strictly modal tools needed for this are the rule of necessitation and the K axiom. But now consider a justification logic analog of $\Box(\tp \wedge \Box \psi)$, say $\lc{t}(\tp \wedge \lc{u} \psi)$. It is not the case that if $\varphi \rightarrow \psi$ is provable in \jlp, so is $\lc{t}(\tp \wedge \lc{u}\varphi) \rightarrow \lc{t}(\tp \wedge \lc{u}\psi)$, but one can easily show that the following is provable,
\begin{equation} 
\label{example1}
\lc{t}(\tp \wedge \lc{u}\varphi) \rightarrow  \lc{[s_{2}\cdot t]}(\tp \wedge \lc{[s_{1} \cdot u]}\psi) \; .
\end{equation}
Here $s_{1}$ is a justification term such that $\lc{s_{1}}(\varphi \rightarrow \psi)$ is provable, provided by the Internalization Theorem~\ref{internalizationtheorem}, and similarly $s_{2}$ is a justification term such that $\lc{s_{2}}((\tp \wedge \lc{u} \varphi) \rightarrow (\tp \wedge \lc{[s_{1} \cdot u]}\psi))$ is provable, again provided by the Internalization Theorem. This is considered in much detail in \cite{Fitting07}, where algorithms are given to compute justification terms involved in such formula manipulations. But for purposes of the present paper we observe the following key fact about (\ref{example1}): both the antecedent and the consequent have the same general form. More specifically, both have the form of $\Box(\tp \wedge \Box \tq)$, with different justification terms replacing the necessity operator. We call an expression like $\Box(\tp \wedge \Box \tq)$ a template and we work with the family of substitution instances of templates, rather than with individual formulas.

Here is another way of looking at the need for templates. In moving from modal to justification logics we can understand the usual modal operator  to be something like an existential quantifier over justification terms (this was made precise and formal in \cite{Fitting08}). Then asserting $\Box \varphi$ is akin to asserting $\lc{t}\varphi$ for some justification term $t$, and asserting $\neg \Box \varphi$ is like asserting $\neg \lc{t}\varphi$ for every justification term $t$. Both of these are infinitary in nature. In the conventional modal setting one talks about $\Gamma \cup \{\neg\Box \varphi\}$ being consistent, for a set $\Gamma$ and formula $\varphi$. The analog in the justification logic setting is talk of the consistency of $\Gamma \cup \{\neg \lc{t_{1}}\varphi, \, \neg \lc{t_{2}}\varphi, \, \dots \}$ where $t_{1}, \, t_{2}, \, \dots$ are all the justification terms. Templates are what we use for this --- we talk about the consistency of $\Gamma$ together with all `instances' of $\neg \Box \varphi$.

\begin{defn} [Template]
Let $\tp_{1}, \, \dots, \, \tp_{n}$ be distinct propositional letters. A \emph{template} on these letters is a propositional modal formula that is built up from $\tp_{1}, \, \dots, \, \tp_{n}$ with $\neg, \vee$ and $\wedge$ as connectives, with  $\Box$ as the only modal operator, and with no $\tp_{i}$ occurring more than once. We write $F(\tp_{1}, \dots , \tp_{n})$ to indicate a template in which the propositional letters are among $\tp_{1}, \dots , \tp_{n}$.  We write $\tvp$ to denote a sequence of propositional variables, so $F(\tp_{1}, \dots , \tp_{n})$ may appear as $F(\tvp)$.
\end{defn}

We emphasize that a template is a propositional modal formula, and not a formula of \folpb\, and hence has no individual variables, quantifiers, or first-order relation symbols.  Just as when we work with formulas, we can define the notion of \emph{degree} of a template as the number of occurrences of boolean connectives and modal operators). Then we can define some notions recursively based on the degree of templates  and prove corresponding facts by induction on the degree of templates.

We next define instantiation sets for templates; these are sets of first-order formulas, Henkin formulas as discussed in Section~\ref{langextsection}.  

\begin{defn}[Instantiation sets] 
Let $\tvp$ be an $n$-ary sequence of distinct propositional variables, $\vvarphi$ be an $n$-ary sequence of Henkin formulas (not necessarily distinct) and $F(\tvp)$ a template. We define the \emph{instantiation set} $ \llbracket F(\vvarphi)  \rrbracket$ recursively as follows: 
\begin{enumerate}[a)]

\item If $F(\tvp)$ is $\tp_{i}$, then $ \llbracket F(\vvarphi)  \rrbracket = \{ \varphi_{i}\}$.

\item If $F(\tvp)$ is $\neg G(\tvp)$, then  $ \llbracket F(\vvarphi)  \rrbracket = \{ \neg \psi \; | \; \psi \in   \llbracket G(\vvarphi)  \rrbracket   \}$.

\item If $F(\tvp)$ is $G(\tvp) \vee H(\tvp)$, then  $ \llbracket F(\vvarphi)  \rrbracket = \{ \psi\vee \theta \; | \; \psi \in   \llbracket G(\vvarphi)  \rrbracket \; \text{and} \; \theta \in   \llbracket H(\vvarphi)  \rrbracket  \}$.

\item If $F(\tvp)$ is $G(\tvp) \wedge H(\tvp)$, then   $ \llbracket F(\vvarphi)  \rrbracket = \{ \psi\wedge \theta \; |\; \psi \in   \llbracket G(\vvarphi)  \rrbracket \; \text{and} \; \theta \in   \llbracket H(\vvarphi)  \rrbracket  \}$.

\item If $F(\tvp)$ is  $\Box G(\tvp)$, then\\   $ \llbracket F(\vvarphi)  \rrbracket = \{ \lc{t}_{X}\psi \; | \; \psi \in   \llbracket G(\vvarphi)  \rrbracket,  t \in \Tj, \mbox{ and } X \mbox{ is the set of witness variables occurring in } \psi\}$.
\end{enumerate}
\end{defn}

Clearly, for every template $F(\tvp)$ and every sequence $\vvarphi$ of Henkin formulas, $\llbracket F(\vvarphi)  \rrbracket$ is a set of Henkin formulas.  

\begin{defn}
We say the template $F$ is \emph{positive} if the only boolean connectives in $F$ are $\land$ and $\lor$ (no $\lnot$), and $F$ is \emph{disjunctive} if the only boolean connective that occurs in $F$ is $\lor$ (no $\land$, no $\lnot$).          
\end{defn}

The rest of this section is devoted to the proof of some fundamental facts about templates. Throughout we assume that there is a fixed variant closed and axiomatically appropriate constant specification $\C$ for the basic language, and that $\Cv$ is its extension. To keep things simple, we will not refer to this assumption in every proposition and, in this section only, we shall write `$\turn$' to denote `$\turncv$', `consistent' to denote `$\Cv$-consistent', `inconsistent' to denote `$\Cv$-inconsistent' and `maximal-consistent' to denote `$\Cv$-maximal consistent'.

\begin{pro} [Semi-replacement]
\label{Semi}
Let $F(\tvp,\tq)$ be a positive template, $\chi$ and $\psi$ Henkin formulas, and $\vvarphi$ a sequence of Henkin formulas. If $\turn \chi \rightarrow \psi$, then for every $\phi \in  \llbracket F(\vvarphi,\chi)   \rrbracket$ there is a $\theta \in  \llbracket F(\vvarphi,\psi)   \rrbracket$ such that $\turn \phi \rightarrow \theta$.   
\end{pro}

\begin{proof} By induction on the degree of $F(\tvp,\tq)$.
\smallskip

Suppose first that $F(\tvp,\tq)$ is atomic. There are two cases to consider.

\noindent i) $F(\tvp,\tq) = \tp_{i}$ where $\tp_i$ is in $\vec{\tp}$. Then $ \llbracket F(\vvarphi,\chi)   \rrbracket =\llbracket F(\vvarphi,\psi)   \rrbracket = \{ \varphi_{i}\}$ so $\phi \rightarrow \theta$ will be $\varphi_i\rightarrow\varphi_i$.

\noindent ii) $F(\tvp,\tq) = \tq$. Then $\llbracket F(\vvarphi,\chi)   \rrbracket = \{ \chi\}$ and $\llbracket F(\vvarphi,\psi)   \rrbracket = \{ \psi\}$, so $\phi\rightarrow\theta$ will be $\chi\rightarrow\psi$. 
\smallskip

Next assume that $F(\tvp,\tq)$ is $G(\tvp, \tq)\vee H(\tvp, \tq)$ and the result is known for $G(\tvp, \tq)$ and $H(\tvp, \tq)$. Let $\phi \in  \llbracket F(\vvarphi,\chi)   \rrbracket$. Then $\phi$ is $\phi\p \vee \phi\pp$ where  $\phi\p \in  \llbracket G(\vvarphi,\chi)   \rrbracket$ and  $\phi\pp \in  \llbracket H(\vvarphi,\chi)   \rrbracket$. By the induction hypothesis, there are $\theta\p \in  \llbracket G(\vvarphi,\psi)   \rrbracket$ and  $\theta\pp \in  \llbracket H(\vvarphi,\psi)   \rrbracket$ such that $\turn \phi\p \rightarrow \theta\p$ and $\turn \phi\pp \rightarrow \theta\pp$ and hence $\turn \phi\p \vee \phi\pp  \rightarrow \theta\p \vee \theta\pp$.  Since $\theta\p \vee \theta\pp \in  \llbracket F(\vvarphi,\psi)   \rrbracket$, take $\theta$ as $\theta\p \vee \theta\pp$. 
\smallskip

If $F(\tvp,\tq)$ is $G(\tvp, \tq)\wedge H(\tvp, \tq)$, then the argument is similar to the previous one.
\smallskip

Suppose $F(\tvp,\tq)$ is $\Box G(\tvp, \tq)$ and the result is known for $G(\tvp, \tq)$. Let $\phi \in  \llbracket F(\vvarphi,\chi)   \rrbracket$; then $\phi$ is $\lc{t}_{X}\phi\p$ where $\phi\p \in  \llbracket G(\vvarphi,\chi)   \rrbracket$. By the induction hypothesis, there is some $\theta\p \in  \llbracket G(\vvarphi,\psi)   \rrbracket$ such that $\turn \phi\p \rightarrow \theta\p$. By Internalization, \ref{internalizationtheorem}, there is a justification term $s \in \Tj$ such that $\turn \lc{s}(\phi\p \rightarrow \theta\p)$.  By repeated use of axiom \textbf{A3} and classical reasoning $\turn \lc{s}_{X}(\phi\p \rightarrow \theta\p)$.  By axiom \textbf{B2} and modus ponens, $\turn \lc{t}_{X}\phi\p \rightarrow \lc{[s\cdot t]}_{X}\theta\p$.  Let $Y$ be the set of all witness variables that occur in $\theta\p$. By repeated use of axioms \textbf{A2} and \textbf{A3}, we have that $\turn \lc{[s\cdot t]}_{X}\theta\p \rightarrow \lc{[s\cdot t]}_{Y}\theta\p$, and hence $\turn \lc{t}_{X}\phi\p \rightarrow \lc{[s\cdot t]}_{Y}\theta\p$.   
Since $\lc{[s\cdot t]}_{Y} \theta\p \in  \llbracket F(\vvarphi,\psi)   \rrbracket$, take $\theta$ to be $\lc{[s\cdot t]}_{Y} \theta\p$. 
\end{proof}

\begin{coro}[Variable change]
\label{variable}
Let $\Gamma \subseteq \Fjv$, $F(\tvp,\tq)$ be a positive template, $\vvarphi$ be a sequence of Henkin formulas, $\forall x \varphi(x)$ be a  Henkin formula, and $y$ be a basic variable that does not occur free or bound in $\forall x \varphi(x)$. If $\Gamma \cup \llbracket \neg F(\vvarphi,\forall x\varphi(x))   \rrbracket$ is consistent, then $\Gamma \cup \llbracket \neg F(\vvarphi,\forall y\varphi(y)) \rrbracket$ is consistent.     
\end{coro}

\begin{proof} Suppose $\Gamma \cup \llbracket \neg F(\vvarphi,\forall y\varphi(y)) \rrbracket$ is inconsistent. Then there are $\lnot\psi_{1}, \dots, \lnot\psi_{n} \in \llbracket \neg F(\vvarphi,\forall y\varphi(y)) \rrbracket$ such that $\Gamma \turn \psi_{1} \vee \, \dots\,\vee \psi_{n}$.  By classical logic, $\turn \forall y\varphi(y) \rightarrow \forall x\varphi(x)$.  Since $\psi_1, \ldots, \psi_n\in \llbracket  F(\vvarphi,\forall y\varphi(y))   \rrbracket$, by Proposition \ref{Semi}, for each $\psi_{i}$ there is some $\theta_{i} \in  \llbracket  F(\vvarphi,\forall x\varphi(x)) \rrbracket$ such that $\turn \psi_{i} \rightarrow \theta_{i}$, and hence $\Gamma \turn \theta_{1} \vee \, \dots\, \vee \theta_{n}$.    Since each $\neg \theta_{i} \in  \llbracket  \neg F(\vvarphi,\forall x\varphi(x)) \rrbracket$, $\Gamma \cup \llbracket  \neg F(\vvarphi,\forall x\varphi(x)) \rrbracket$ is inconsistent.
\end{proof}

\begin{pro}[Vacuous quantification]
\label{vacuous}
Let $F(\tvp)$ be a disjunctive template, and $\vvarphi$ be a sequence of Henkin formulas none of which contain free occurrences of the basic variable $y$. For each $\psi \in  \llbracket F(\vvarphi)  \rrbracket$ there is some $\theta \in  \llbracket F(\vvarphi)  \rrbracket$ such that $\turn \exists y\psi \rightarrow \theta$.
\end{pro}

\begin{proof} By induction on the degree of $F(\tvp)$.
\smallskip

Suppose $F(\tvp)$ is $\tp_{i}$.  Then $\llbracket F(\vvarphi)  \rrbracket =\{ \varphi_{i}  \}$. Since $y$ does not occur free in $\varphi_{i}$, $\turn \exists y\varphi_{i} \rightarrow \varphi_{i}$, so we take $\theta$ to be $\varphi_{i}$.
\smallskip

Next assume that $F(\tvp)$ is $G(\tvp)\vee H(\tvp)$ and the result is known for $G(\tvp)$ and $ H(\tvp)$. Let $\psi \in  \llbracket F(\vvarphi)  \rrbracket$. Then $\psi$ is $\psi\p \vee \psi\pp$ where  $\psi\p \in  \llbracket G(\vvarphi)  \rrbracket$ and  $\psi\pp \in  \llbracket H(\vvarphi)  \rrbracket$. By the induction hypothesis there are $\theta\p \in  \llbracket G(\vvarphi)  \rrbracket$ and  $\theta\pp \in  \llbracket H(\vvarphi)  \rrbracket$ such that $\turn \exists y \psi\p \rightarrow \theta\p \; and \; \turn \exists y\psi\pp \rightarrow \theta\pp$.  By classical logic, $\turn \exists y (\psi\p \vee \psi\pp) \longleftrightarrow  (\exists y\psi\p \vee \exists y\psi\pp)$, hence $\turn \exists y (\psi\p \vee \psi\pp) \rightarrow \theta\p \vee \theta\pp$.      
Since $\theta\p \vee \theta\pp \in  \llbracket F(\vvarphi)  \rrbracket$, take $\theta$ to be $\theta\p \vee \theta\pp$.
\smallskip

Finally suppose $F(\tvp)$ is $\Box G(\tvp)$ and the result is known for $ G(\tvp)$. Let $\psi \in  \llbracket F(\vvarphi)  \rrbracket$. So $\psi$ is $\lc{t}_{X}\phi$ where $\phi \in  \llbracket G(\vvarphi)  \rrbracket$. By the axiom \textbf{A3}, $\turn \lc{t}_{X}\phi \rightarrow \lc{t}_{Xy}\phi$.  Then using classical logic, $\turn \exists y \lc{t}_{X}\phi \rightarrow \exists y \lc{t}_{Xy}\phi$.  By definition, $X$ is a set of witness variables and since $y$ is a basic variable we have that $y \notin X$; so by 2.\ of Theorem \ref{cbarcanburidan}, $\turn \exists y \lc{t}_{Xy}\phi \rightarrow \lc{s(t)}_{X}\exists y \phi$.

By the induction hypothesis, there is some $\theta\p \in  \llbracket G(\vvarphi)  \rrbracket$ such that $\turn \exists y \phi \rightarrow \theta\p$. By Internalization, \ref{internalizationtheorem}, there is a justification term $s\p \in \Tj$ such that $\turn \lc{s\p}_{X}(\exists y \phi \rightarrow \theta\p)$.  Using axiom \textbf{B2}, $\turn \lc{s(t)}_{X}\exists y \phi \rightarrow \lc{[s\p \cdot s(t)]}_{X}\theta\p$.  Let $Y$ be the set of all witness variables that occur in $\theta\p$. By repeated use of axioms \textbf{A2} and \textbf{A3}, we have $\turn \lc{[s\p \cdot s(t)]}_{X}\theta\p \rightarrow \lc{[s\p \cdot s(t)]}_{Y}\theta\p$, hence $\turn \exists y \lc{t}_{X}\phi \rightarrow \lc{[s\p \cdot s(t)]}_{Y}\theta\p$.  Since $\lc{[s\p \cdot s(t)]}_{Y}\theta\p \in  \llbracket F(\vvarphi)  \rrbracket$, we take $\theta$ to be $\lc{[s\p \cdot s(t)]}_{Y}\theta\p$.   
\end{proof}

\begin{pro}[Generalized Barcan]
\label{generalized}
Let $F(\tvp,\tq)$ be a disjunctive template, $y$ a basic variable, $\varphi(y)$ a Henkin formula, and  $\vvarphi$ a sequence of Henkin formulas none of which contain free occurrences of $y$. For each $\psi \in \llbracket F(\vvarphi,\varphi(y)) \rrbracket$ there is some $\theta \in \llbracket F(\vvarphi,\forall y \varphi(y)) \rrbracket$ such that $\turn \forall y\psi \rightarrow \theta$.
\end{pro}

\begin{proof} By induction on the degree of $F(\tvp,\tq)$.
\smallskip

If $F(\tvp,\tq)$ is atomic, then the result is trivial.
\smallskip

Next assume that $F(\tvp,\tq)$ is $G(\tvp, \tq)\vee H(\tvp, \tq)$ and the result is known for $G(\tvp, \tq)$ and $H(\tvp, \tq)$. By the definition of template, the propositional variable $\tq$ can occur at most once in $F(\tvp,\tq)$, so either it does not occur in $G(\tvp,\tq)$ or it does not occur in  $H(\tvp,\tq)$. Assume that it does not occur in $H(\tvp,\tq)$ (the other case has a similar proof); then we can write $H(\tvp,\tq)$ as $H(\tvp)$.    

Let $\psi \in \llbracket F(\vvarphi,\varphi(y)) \rrbracket$. So $\psi$ is $\phi\p \vee \phi\pp$ where  $\phi\p \in \llbracket G(\vvarphi,\varphi(y)) \rrbracket$ and  $\phi\pp \in  \llbracket H(\vvarphi)  \rrbracket$. By classical logic we have $\turn \forall y (\phi\p \vee \phi\pp) \rightarrow (\forall y \phi\p \vee \exists y \phi\pp)$.  Since $y$ does not occur free in any formula of $\vvarphi$, then by Proposition \ref{vacuous} there is some $\theta\pp \in  \llbracket H(\vvarphi)  \rrbracket$ such that $\turn \exists y \phi\pp \rightarrow \theta\pp$.  By the induction hypothesis, there is some $\theta\p \in \llbracket G(\vvarphi,\forall y \varphi (y)) \rrbracket$ such that $\turn \forall y \psi\p \rightarrow \theta\p$ and hence $\turn \forall y (\phi\p \vee \phi\pp) \rightarrow \theta\p \vee \theta\pp$.  Then we can take $\theta$ as $\theta\p \vee \theta\pp$.
\smallskip

Suppose $F(\tvp,\tq)$ is $\Box G(\tvp,\tq)$ and the result is known for $G(\tvp,\tq)$. Let $\psi \in \llbracket F(\vvarphi, \varphi(y)) \rrbracket$. Then $\psi$ is $\lc{t}_{X}\phi$ where $\phi \in \llbracket G(\vvarphi,\varphi(y)) \rrbracket$. By axiom \textbf{Bb}, $\turn \forall y \lc{t}_{Xy}\phi \rightarrow \lc{\jb(t)}_{X}\forall y \phi$. By axiom \textbf{A3}, $\turn \lc{t}_{X}\phi \rightarrow \lc{t}_{Xy}\phi$ so by classical logic, $\turn \forall y\lc{t}_{X}\phi \rightarrow \forall y \lc{t}_{Xy}\phi$.  Then $\turn \forall y\lc{t}_{X}\phi \rightarrow \lc{\jb(t)}_{X}\forall y \phi$.
    
By the induction hypothesis, there is some $\theta\p \in \llbracket G(\vvarphi, \forall y \varphi(y)) \rrbracket$ such that $\turn \forall y \phi \rightarrow \theta\p$. By Internalization, \ref{internalizationtheorem}, and axiom \textbf{A3}, there is a justification term $s \in \Tj$ such that $\turn \lc{s}_{X}(\forall y \phi \rightarrow \theta\p)$.  Then by axiom \textbf{B2}, $\turn \lc{\jb(t)}_{X}\forall y \phi \rightarrow \lc{[s \cdot \jb(t)]}_{X}\theta\p$. 
Let $Y$ be the set of all witness variables that occur in $\theta\p$. By repeated use of axioms \textbf{A2} and \textbf{A3}, we have that $\turn \lc{[s \cdot \jb(t)]}_{X}\theta\p \rightarrow \lc{[s \cdot \jb(t)]}_{Y}\theta\p$ and hence $\turn \forall y\lc{t}_{X}\phi \rightarrow \lc{[s \cdot \jb(t)]}_{Y}\theta\p$.  Take $\theta$ to be $\lc{[s \cdot \jb(t)]}_{Y}\theta\p$.    
\end{proof}

\begin{pro}[Formula combining]
\label{combining}
Let $F(\tvp)$ be a disjunctive template, and $\vvarphi$ be a sequence of Henkin formulas. For any $\psi_{1}, \, \dots, \, \psi_{k}  \in  \llbracket F(\vvarphi)  \rrbracket$ there is a formula $\theta \in  \llbracket F(\vvarphi)  \rrbracket$ such that $\turn (\psi_{1} \vee \, \dots\, \vee \psi_{k}) \rightarrow \theta$.
\end{pro}

\begin{proof} Induction on the degree of $F(\tvp)$.
\smallskip

If $F(\tvp)$ is atomic, the result is trivial.
\smallskip

Assume that $F(\tvp)$ is $G(\tvp)\vee H(\tvp)$ and the result is known for $G(\tvp)$ and $H(\tvp)$. Let $\psi_{1}, \, \dots, \, \psi_{k} \in  \llbracket F(\vvarphi)  \rrbracket$. Then there are $\phi_{1}^{\p}, \, \dots, \, \phi_{k}^{\p} \in  \llbracket G(\vvarphi)  \rrbracket$ and  $\phi_{1}^{\pp}, \, \dots, \, \phi_{k}^{\pp} \in  \llbracket H(\vvarphi)  \rrbracket$, such that $\psi_{i} = \phi_{i}^{\p} \vee \phi_{i}^{\pp}$. By the induction hypothesis, there are $\theta^{\p} \in  \llbracket G(\vvarphi)  \rrbracket$ and $\theta^{\pp} \in  \llbracket H(\vvarphi)  \rrbracket$ such that $\turn (\phi_{1}^{\p} \vee \, \dots \, \vee \phi_{k}^{\p}) \rightarrow \theta\p$ and $\turn (\phi_{1}^{\pp} \vee \, \dots \,\vee \phi_{k}^{\pp}) \rightarrow \theta\pp$.  Then $\turn ((\phi_{1}^{\p} \vee\, \dots \,\vee \phi_{k}^{\p}) \vee (\phi_{1}^{\pp} \vee\, \dots \,\vee \phi_{k}^{\pp})) \rightarrow \theta\p \vee \theta\pp$ and by rearranging, $\turn ((\phi_{1}^{\p} \vee \phi_{1}^{\pp})  \vee\, \dots \,\vee (\phi_{k}^{\p} \vee \phi_{k}^{\pp}))  \rightarrow \theta\p \vee \theta\pp$.  That is, $\turn (\psi_{1} \vee\, \dots \,\vee \psi_{k}) \rightarrow \theta\p \vee \theta\pp$, so take $\theta$ to be $\theta\p \vee \theta\pp$.
\smallskip

Assume that $F(\tvp)$ is $\Box G(\tvp)$ and the result is known for $G(\tvp)$. Let $\psi_{1}, \, \dots, \, \psi_{k} \in  \llbracket F(\vvarphi)  \rrbracket$. Then there are justification terms $t_{1},\, \dots, \, t_{k}$ and $\phi_{1},\, \dots, \, \phi_{k} \in  \llbracket G(\vvarphi)  \rrbracket$ such that $\psi_{i} = \lc{t_{i}}_{X_{i}}\phi_{i}$. By the induction hypothesis, there is some $\theta^{\p} \in  \llbracket G(\vvarphi)  \rrbracket$ such that $\turn (\phi_{1} \vee\, \dots \,\vee \phi_{k}) \rightarrow \theta\p$.

By classical reasoning, for each $i$, $\turn \phi_{i} \rightarrow \theta\p$.  Then by Internalization, \ref{internalizationtheorem}, and axiom \textbf{A3} there are  justification terms $s_{1}, \, \dots, \, s_{k} \in \Tj$ such that for each $i$, $\turn \lc{s_{i}}_{X_{i}}(\phi_{i} \rightarrow \theta\p)$.  Using axiom \textbf{B2}, $\turn \lc{t_{i}}_{X_{i}}\phi_{i} \rightarrow \lc{[s_{i} \cdot t_{i}]}_{X_{i}}\theta\p$, and then by axiom \textbf{B3} we have that for each $i$, $\turn \lc{[s_{i} \cdot t_{i}]}_{X_{i}}\theta\p \rightarrow \lc{[[s_{1} \cdot t_{1}] + \, \dots \, + [s_{k} \cdot t_{k}]]}_{X_{i}}\theta\p$.

Let $Y$ be the set of all witness variables that occur in $\theta\p$. By repeated use of axioms \textbf{A2} and \textbf{A3}, we have $\turn \lc{[[s_{1} \cdot t_{1}] + \, \dots \, + [s_{k} \cdot t_{k}]]}_{X_{i}}\theta\p \rightarrow \lc{[[s_{1} \cdot t_{1}] + \, \dots \, + [s_{k} \cdot t_{k}]]}_{Y}\theta\p$, and so for each $i$, $\turn \lc{t_{i}}_{X_{i}}\phi_{i} \rightarrow \lc{[[s_{1} \cdot t_{1}] + \, \dots \, + [s_{k} \cdot t_{k}]]}_{Y}\theta\p$.  Then $\turn (\lc{t_{1}}_{X_{1}}\phi_{1} \vee \,  \dots \,  \vee \lc{t_{k}}_{X_{k}}\phi_{k}) \rightarrow \lc{[[s_{1} \cdot t_{1}] + \, \dots \, + [s_{k} \cdot t_{k}]]}_{Y}\theta\p$.  Since $\lc{[[s_{1} \cdot t_{1}] + \, \dots \, + [s_{k} \cdot t_{k}]]}_{Y}\theta\p \in  \llbracket F(\vvarphi)  \rrbracket$, we can take $\theta$ to be $\lc{[[s_{1} \cdot t_{1}] + \, \dots \, + [s_{k} \cdot t_{k}]]}_{Y}\theta\p$.
\end{proof}

\begin{pro}[Existential Instantiation]
\label{existentialinstantiation}    
Let $F(\tvp,\tq)$ be a disjunctive template, $\Gamma \subseteq \Fj$, $\vec{\chi}$ be a sequence of Henkin formulas,  $\forall x \varphi(x)$ be a Henkin formula, and $a$ be a witness variable that does not occur free in $\forall x \varphi(x)$ or in any member of $\vec{\chi}$. If $\Gamma \cup  \llbracket \neg F(\vec{\chi},\forall x\varphi(x))  \rrbracket$ is consistent, then $\Gamma \cup  \llbracket \neg F(\vec{\chi},\varphi(a))  \rrbracket$ is consistent.    
\end{pro}

\begin{proof} 
We show the contrapositive.  Suppose $\Gamma \cup  \llbracket \neg F(\vec{\chi},\varphi(a))  \rrbracket$ is inconsistent.  We show that  $\Gamma \cup  \llbracket \neg F(\vec{\chi},\forall x\varphi(x))  \rrbracket$ is inconsistent.

Assuming $\Gamma \cup  \llbracket \neg F(\vec{\chi},\varphi(a))  \rrbracket$ is inconsistent, there are $\psi_{1}, \, \dots, \, \psi_{n} \in \Gamma$ and $\neg \phi_{1}(a), \, \dots, \, \neg \phi_{k}(a) \in  \llbracket \neg F(\vec{\chi},\varphi(a))  \rrbracket$ such that $\turn(\psi_{1} \wedge \, \dots\, \wedge \psi_{n}) \wedge (\neg \phi_{1}(a) \wedge \, \dots \, \wedge \neg \phi_{k}(a)) \rightarrow \bot$.  It follows that $\turn(\psi_{1} \wedge \, \dots\, \wedge \psi_{n}) \rightarrow (\phi_{1}(a) \vee\, \dots \, \vee \phi_{k}(a))$.  By Proposition~\ref{combining} there is some $\psi(a) \in  \llbracket F(\vec{\chi},\varphi(a))  \rrbracket$ such that $\turn (\phi_{1}(a) \vee\, \dots \, \vee \phi_{k}(a)) \rightarrow \psi(a)$ and hence $\turn (\psi_{1} \wedge\, \dots \,\wedge \psi_{n}) \rightarrow \psi(a)$.  Note that since $\Gamma\subseteq \Fj$ is a set of basic formulas, $a$ does not occur in any formula of $\Gamma$.

By Corollary~\ref{gencorollary}, $\turn \forall y [(\psi_{1} \wedge\, \dots \,\wedge \psi_{n}) \rightarrow \psi(y)]$, where $y$ is a new basic variable replacing all (free) occurrences of witness variable $a$.  And since there were no occurrences of $a$ in any $\psi_i$, after substitution there are no occurrences of $y$, hence $\turn (\psi_{1} \wedge\, \dots \,\wedge \psi_{n}) \rightarrow \forall y \psi(y)$.
    
Since $a$ does not occur in any formula of $\vec{\chi}$, it can be easily checked that for every formula $\gamma(a)$: $\text{if} \; \gamma(a) \in  \llbracket F(\vec{\chi},\varphi(a))  \rrbracket, \, \text{then} \; \gamma(y) \in  \llbracket F(\vec{\chi},\varphi(y))  \rrbracket$.  Then $\psi(y) \in  \llbracket F(\vec{\chi},\varphi(y))  \rrbracket$, and since $y$ is new it does not occur in $\vec{\chi}$, by Proposition \ref{generalized} there is a $\theta \in  \llbracket F(\vec{\chi},\forall y \varphi(y))  \rrbracket$ such that $\turn \forall y\psi (y)\rightarrow \theta$, and thus $\turn  (\psi_{1} \wedge\, \dots\, \wedge \psi_{n})  \rightarrow \theta$.  Since $\neg\theta \in  \llbracket \neg F(\vec{\chi},\forall y \varphi(y))  \rrbracket$ but $\psi_1, \ldots, \psi_n\in\Gamma$, it follows that $\Gamma \cup \llbracket \neg F(\vec{\chi},\forall y\varphi(y))  \rrbracket$ is inconsistent, and by Corollary \ref{variable}, $\Gamma \cup  \llbracket \neg F(\vec{\chi},\forall x\varphi(x))  \rrbracket$ is inconsistent.
\end{proof}

In a propositional modal setting, a way of formulating the standard definition of the accessibility relation is this.  Possible worlds are maximally consistent sets.  For a possible world $\Gamma$, define $\Gamma^\# = \{\varphi \mid \Box\varphi\in\Gamma\}$.  Then, for two possible worlds, set $\Gamma\R\Delta$ if $\Gamma^\#\subseteq\Delta$.  The following will give us the appropriate analog for quantified justification logic models.

\begin{defn}
\label{sharp}
Suppose $\Gamma \subseteq \Fjv$.  Let $\Gamma^{\#}$ be the set of all formulas $\forall \vec{y} \varphi$ such that $\lc{t}_{X}\varphi \in \Gamma$ for some closed Henkin formula $\lc{t}_{X}\varphi$ in which $X$ is the set of witness variables in $\varphi $, and $\vec{y}$ are the free  basic variables of $\varphi$. 
\end{defn}

\begin{pro}[Up and Down Consistency]
\label{up}
Let $F(\tvp) = \Box G(\tvp)$ be a template, $\Gamma \subseteq \Fjv$,  and $\vvarphi$ a sequence of Henkin formulas.
\begin{enumerate}[1.]
\item Suppose $\Gamma$ is maximally consistent. If $\Gamma^{\#} \cup \llbracket  \neg G(\vvarphi) \rrbracket$ is consistent, so is $\Gamma \cup \llbracket \neg F(\vvarphi) \rrbracket$.
\item Suppose $G(\tvp)$ is a disjunctive template. If $\Gamma \cup \llbracket  \neg F(\vvarphi) \rrbracket$ is consistent, so is $\Gamma^{\#} \cup \llbracket  \neg G(\vvarphi) \rrbracket$.
\end{enumerate}    
\end{pro}

\begin{proof} Part 1: Suppose $\Gamma$ is maximally consistent but $\Gamma \cup \llbracket  \neg F(\vvarphi) \rrbracket$ is inconsistent.  We show $\Gamma^{\#} \cup \llbracket \neg G(\vvarphi) \rrbracket$ is inconsistent.

Since $\Gamma \cup \llbracket  \neg F(\vvarphi) \rrbracket$ is inconsistent, for some $\neg \lc{t_{1}}_{X_{1}}\theta_{1}, \, \dots, \,  \neg \lc{t_{k}}_{X_{k}}\theta_{k} \in  \llbracket  \neg F(\vvarphi) \rrbracket$ we have that $\Gamma \cup\{\neg \lc{t_{1}}_{X_{1}}\theta_{1}, \ldots,  \neg \lc{t_{k}}_{X_{k}}\theta_{k}\}\turn  \bot$.  (Note that $\theta_{1}, \, \dots, \, \theta_{k} \in  \llbracket   G(\vvarphi) \rrbracket$.)  Using the Deduction Theorem, $\Gamma \turn (\neg \lc{t_{1}}_{X_{1}}\theta_{1} \wedge \, \dots \, \wedge  \neg \lc{t_{k}}_{X_{k}}\theta_{k}) \rightarrow \bot$, and hence $\Gamma \turn \lc{t_{1}}_{X_{1}}\theta_{1}  \vee \, \dots \, \vee   \lc{t_{k}}_{X_{k}}\theta_{k}$.

Since $\Gamma$ is a maximally consistent set, for some $i$, $t_{i}$$:_{X_{i}}$$\theta_{i} \in \Gamma$. And since $\lc{t_{i}}_{X_{i}}\theta_{i}$ is a closed Henkin formula, $\forall \vec{x} \theta_{i} \in \Gamma^{\#}$. By classical logic,   $\turn \neg \theta_{i} \rightarrow \neg \forall \vec{x} \theta_{i}$.    
Since $\neg \theta_{i} \in \llbracket \neg G(\vvarphi) \rrbracket$, we have that  $\Gamma^{\#} \cup \llbracket  \neg G(\vvarphi) \rrbracket$ is inconsistent.
\medskip

Part 2: Suppose $\Gamma^{\#} \cup \llbracket  \neg G(\vvarphi) \rrbracket$ is inconsistent.  We show that $\Gamma \cup \llbracket  \neg F(\vvarphi) \rrbracket$ is inconsistent.

Assume $\Gamma^{\#} \cup \llbracket  \neg G(\vvarphi) \rrbracket$ is inconsistent. Then there are $\forall \vec{x}_{1} \psi_{1}, \dots, \forall \vec{x}_{n} \psi_{n} \in \Gamma^{\#}$ and $\neg \theta_{1}, \dots, \neg \theta_{k} \in \llbracket   \neg G(\vvarphi) \rrbracket$ such that $\turn (\forall \vec{x}_{1} \psi_{1} \wedge \, \dots \, \wedge \forall \vec{x}_{n} \psi_{n}) \wedge  (\neg \theta_{1}\wedge \, \dots \, \wedge \neg \theta_{k}) \rightarrow \bot$, and so $\turn (\forall \vec{x}_{1} \psi_{1} \wedge \, \dots \, \wedge \forall \vec{x}_{n} \psi_{n}) \rightarrow  (\theta_{1} \vee \, \dots \, \vee  \theta_{k})$.  Note that since $\forall \vec{x}_{1} \psi_{1}, \dots, \forall \vec{x}_{n} \psi_{n} \in \Gamma^{\#}$, by definition of $\Gamma^\#$ there must be corresponding $\lc{t_{1}}_{X_{1}}\psi_{1}, \dots,  \lc{t_{n}}_{X_{n}}\psi_{n} \in  \Gamma$.

Since $\theta_{1}, \,  \dots, \,  \theta_{k} \in  \llbracket G(\vvarphi)  \rrbracket$ and $G(\tvp)$ is a disjunctive template, by Proposition \ref{combining} there is some $\theta \in  \llbracket G(\vvarphi)  \rrbracket$ such that $\turn (\theta_{1} \vee \, \dots \, \vee  \theta_{k}) \rightarrow \theta$.  Then by classical logic, $\turn \forall \vec{x}_{1} \psi_{1} \rightarrow \, \dots \, \rightarrow \forall \vec{x}_{n} \psi_{n} \rightarrow  \theta$.

For each $i$, no member of the sequence $\vec{x}_{i}$ may occur in the corresponding set $X_{i}$ of witness variables. Then by repeated use of axiom \textbf{B5} we have that for each $i$, $\turn  \lc{t_{i}}_{X_{i}}\psi_{i}\rightarrow \lc{\gen_{\vec{x}_{i}}(t)}_{X_{i}}\forall \vec{x_{i}}\psi_{i}$.  (We noted in Section~\ref{semantics} that $\gen_{\vec{x}_{i}}$ would abbreviate nested occurrences of $\gen$.)  If we let $X = X_{1}\cup \, \dots \, \cup X_{n}$, by using axiom \textbf{A3}, $\turn \lc{t_{i}}_{X_{i}}\psi_{i}\rightarrow \lc{\gen_{\vec{x}_{i}}(t)}_{X}\forall \vec{x_{i}}\psi_{i}$.

We have that $\turn \forall \vec{x}_{1} \psi_{1} \rightarrow \, \dots \, \rightarrow \forall \vec{x}_{n} \psi_{n} \rightarrow  \theta$.  Then by Internalization, \ref{internalizationtheorem}, and axiom \textbf{A3} there is a justification term $s \in \Tj$ such that $\turn \lc{s}_{X}(\forall \vec{x}_{1} \psi_{1} \rightarrow \, \dots \, \rightarrow \forall \vec{x}_{n} \psi_{n} \rightarrow  \theta)$, where $X$ is $X_1\cup\ldots\cup X_n$.  Then by repeated use of axiom \textbf{B2}, $\turn \lc{\gen_{\vec{x}_{1}}(t)}_{X}\forall \vec{x}_{1}\psi_{1}\rightarrow \, \dots \,\rightarrow \gen_{\vec{x}_{n}}(t)_{X} \forall \vec{x}_{n}\psi_{n}\rightarrow \lc{[s\cdot \gen_{\vec{x}_{1}}(t) \cdot \, \dots \, \cdot \gen_{\vec{x}_{n}}(t) ]}_{X}\theta$.  Let $Y$ be the set of those witness variables occurring in $\theta$.  Using axioms \textbf{A2} and \textbf{A3}, $\turn \lc{[s\cdot \gen_{\vec{x}_{1}}(t) \cdot \, \dots \, \cdot \gen_{\vec{x}_{n}}(t) ]}_{X}\theta \rightarrow \lc{[s\cdot \gen_{\vec{x}_{1}}(t) \cdot \, \dots \, \cdot \gen_{\vec{x}_{n}}(t) ]}_{Y}\theta$.  Combining, we have $\turn \lc{\gen_{\vec{x}_{1}}(t)}_{X}\forall \vec{x}_{1}\psi_{1}\rightarrow \, \dots \,\rightarrow \gen_{\vec{x}_{n}}(t)_{X} \forall \vec{x}_{n}\psi_{n}\rightarrow \lc{[s\cdot \gen_{\vec{x}_{1}}(t) \cdot \, \dots \, \cdot \gen_{\vec{x}_{n}}(t) ]}_{Y}\theta$, and hence $\turn [\lc{\gen_{\vec{x}_{1}}(t)}_{X}\forall \vec{x}_{1}\psi_{1}\land \, \dots \,\land \gen_{\vec{x}_{n}}(t)_{X} \forall \vec{x}_{n}\psi_{n}]\rightarrow \lc{[s\cdot \gen_{\vec{x}_{1}}(t) \cdot \, \dots \, \cdot \gen_{\vec{x}_{n}}(t) ]}_{Y}\theta$.  We showed above that for each $i$, $\turn \lc{t_{i}}_{X_{i}}\psi_{i}\rightarrow \lc{\gen_{\vec{x}_{i}}(t)}_{X}\forall \vec{x_{i}}\psi_{i}$.  Then by classical reasoning, $\turn ( \lc{t_{1}}_{X_{1}}\psi_{1} \wedge \, \dots \, \wedge \lc{t_{n}}_{X_{n}}\psi_{n}) \rightarrow \lc{[s\cdot \gen_{\vec{x}_{1}}(t) \cdot \, \dots \, \cdot \gen_{\vec{x}_{n}}(t) ]}_{Y} \theta$.

Finally, each $\lc{t_{i}}_{X_{i}}\psi_{i} \in \Gamma$ so $\Gamma\turn\lc{[s\cdot \gen_{\vec{x}_{1}}(t) \cdot \, \dots \, \cdot \gen_{\vec{x}_{n}}(t) ]}_{Y} \theta$.  Also $\theta \in  \llbracket G(\vvarphi)  \rrbracket$, and $F(\tvp) = \Box G(\tvp)$, so $\lc{[s\cdot \gen_{\vec{x}_{1}}(t) \cdot \, \dots \, \cdot \gen_{\vec{x}_{n}}(t) ]}_{Y} \theta\in \llbracket F(\vvarphi)  \rrbracket$, and hence $\lnot\lc{[s\cdot \gen_{\vec{x}_{1}}(t) \cdot \, \dots \, \cdot \gen_{\vec{x}_{n}}(t) ]}_{Y} \theta\in \llbracket \lnot F(\vvarphi)  \rrbracket$.  It follows that $\Gamma \cup \llbracket  \neg F(\vvarphi) \rrbracket$ is inconsistent.
\end{proof}

\begin{defn} [Admitting instantiation]
A set of formulas $\Gamma$ \emph{admits instantiation} provided that for each disjunctive template $F(\tvp,\tq)$, for each sequence $\vvarphi$ of Henkin formulas, and each universally quantified Henkin formula $\forall x \varphi (x)$, if $\Gamma \cup \llbracket \neg F(\vvarphi,\forall x\varphi(x)) \rrbracket$ is consistent, then for some witness variable $a$, $\Gamma \cup \llbracket \neg F(\vvarphi,\varphi(a)) \rrbracket$ is consistent.
\end{defn}

\begin{pro}
\label{instantiation}
Suppose $\Gamma$ is maximally consistent and admits instantiation. For every universal Henkin formula $\forall x \varphi(x)$, if $\neg \forall x \varphi(x) \in \Gamma$, then there is a witness variable $a$ such that $\neg \varphi(a) \in \Gamma$.          
\end{pro}

\begin{proof} 
If $\neg \forall x \varphi (x) \in \Gamma $, then of course $\Gamma\cup \{ \neg \forall x \varphi (x) \}$ is consistent. Let $\tq$ be a propositional letter; $F(\tq) =\tq$ is trivially a disjunctive template. Since $\llbracket  \neg F(\forall x\varphi(x)) \rrbracket =  \{ \neg \forall x \varphi (x) \}$, then  $\Gamma \cup \llbracket  \neg F(\forall x\varphi(x)) \rrbracket$ is consistent. Since $\Gamma$ admits instantiation, there is a witness variable $a$ such that $\Gamma \cup \llbracket  \neg F(\varphi(a)) \rrbracket$ is consistent, i.e., $\Gamma \cup  \{ \neg \varphi (a) \}$ is consistent. By the maximality of $\Gamma$, $ \neg \varphi(a) \in \Gamma$.   
\end{proof}

\begin{pro}
\label{instantiationsharp}  
Let $\Gamma \subseteq \Fjv$. If $\Gamma$ is maximally consistent and admits instantiation, then $\Gamma^{\#}$ also admits instantiation.
\end{pro}

\begin{proof} 
Suppose $\Gamma$ is maximally consistent, $\Gamma$ admits instantiation, $F(\tvp, \tq)$ is a disjunctive template, $\vvarphi$ is a sequence of Henkin formulas, $\forall x \varphi(x)$ is a Henkin formula, and $\Gamma^{\#}\cup \llbracket  \neg F(\vvarphi,\forall x\varphi(x)) \rrbracket$ is consistent. By item 1. of Proposition \ref{up}, $\Gamma\cup \llbracket  \neg\Box F(\vvarphi,\forall x\varphi(x)) \rrbracket$ is consistent.  $\Box F(\tvp, \tq)$ is also a disjunctive template. Then, since $\Gamma$ admits instantiation, for some witness variable $a$,  $\Gamma\cup \llbracket  \neg\Box F(\vvarphi,\varphi(a)) \rrbracket$ is consistent. By item 2. of Proposition \ref{up}, $\Gamma^{\#}\cup \llbracket  \neg F(\vvarphi,\varphi(a)) \rrbracket$ is consistent.       
\end{proof}

\section{Using Templates for Henkin-Like Constructions}\label{henkinsection}

The set of all disjunctives templates is countable, as are $\Fjv$, and the set of all finite sequences $\vvarphi$ of Henkin formulas. Hence the set of all pairs $\bl F(\tvp), \vvarphi \br$ is countable, where $F$ is a disjunctive template, $\tvp$ is an $n$-ary sequence of propositional variables and $\vvarphi$ is an $n$-ary sequence of Henkin formulas.  For the rest of this section we shall assume that the members of the set of all such disjunctive template, formula pairs $\bl F(\tvp), \vvarphi \br$ is arranged in a fixed sequence which we will call the \emph{enumerated sequence}.
\[
\bl F_{1}(\tvp_{1}), \vvarphi_{1} \br, \, \bl F_{2}(\tvp_{2}), \vvarphi_{2} \br, \, \bl F_{3}(\tvp_{3}), \vvarphi_{3} \br, \, \dots 
\]
The enumerated sequence determines a corresponding sequence of \emph{instantiation sets of the enumerated sequence}:
\[
\llbracket  F_{1}( \vvarphi_{1}) \rrbracket, \,  \llbracket  F_{2}( \vvarphi_{2}) \rrbracket, \,  \llbracket F_{3}( \vvarphi_{3}) \rrbracket , \,  \dots 
\]
It should be noted that for two different pairs $\bl F_{i}(\tvp_{i}), \vvarphi_{i} \br$, $\bl F_{j}(\tvp_{j}), \vvarphi_{j} \br$ in the enumerated sequence, the corresponding instantiation sets may be the same. For example, the pairs $\bl \tp_{0}, \bl \forall x \varphi (x)\br\br$ and $\bl \tp_{1}, \bl \forall x \varphi (x)\br \br$ determine the same set $\{\forall x \varphi(x)\}$. This is actually useful to us.  If $\llbracket  F_{i}( \vvarphi_{i}) \rrbracket$ is an instantiation set, instantiating $\llbracket  F_{i}( \vvarphi_{i}) \rrbracket$, we can always find a different member $\llbracket  F_{j}( \vvarphi_{j}) \rrbracket$ of the enumerated sequence \emph{with different propositional variables}, having the same instantiation set.

\begin{lemma}
\label{beforeBe}
Let $\C$ be a variant closed and axiomatically appropriate constant specification for the basic language, $\Cv$ its extension, and $\Gamma \subseteq \Fj$. For any finite union of instantiation sets of the enumerated sequence $\llbracket  F_{i_{1}}( \vvarphi_{i_{1}}) \rrbracket  \cup \, \dots \, \cup \llbracket F_{i_{n}}( \vvarphi_{i_{n}}) \rrbracket$, for any disjunctive template $G(\vec{\tq},\tr)$, and for any Henkin formulas $\vec{\psi}, \forall x \varphi(x)$ if   $\Gamma \cup \llbracket  \neg F_{i_{1}}( \vvarphi_{i_{1}}) \rrbracket  \cup \, \dots \, \cup \llbracket  \neg F_{i_{n}}( \vvarphi_{i_{n}}) \rrbracket \, \cup \llbracket  \neg G( \vec{\psi}, \forall x\varphi(x) \rrbracket $ is $\Cv$-consistent, then there is a witness variable $a$ such that $\Gamma \cup \llbracket  \neg F_{i_{1}}( \vvarphi_{i_{1}}) \rrbracket  \cup \, \dots \, \cup \llbracket  \neg F_{i_{n}}( \vvarphi_{i_{n}}) \rrbracket  \cup \llbracket  \neg G(\vec{\psi}, \varphi (a)) \rrbracket$  is $\Cv$-consistent.
\end{lemma}   

\begin{proof}
We assume that $\llbracket  F_{i_{1}}( \vvarphi_{i_{1}}) \rrbracket$, \ldots, $\llbracket F_{i_{n}}( \vvarphi_{i_{n}}) \rrbracket$ instantiate $\bl F_{i_1}(\tvp_{i_1}), \vvarphi_{i_1} \br$, \ldots $\bl F_{i_n}(\tvp_{i_n}), \vvarphi_{i_n} \br$ and, making use of the remarks above, there is no repetition among the propositional variables $\tvp_{i_{1}}, \,  \dots, \, \tvp_{i_{n}}$, $\vec{\tq}$, $\tr$. Then $F_{i_{1}}(\tvp_{i_{1}})\vee \, \dots\, \vee F_{i_{n}}(\tvp_{i_{n}}) \vee G(\vec{\tq},\tr)$ is a disjunctive template.

From the definition of instantiation set, using classical reasoning, it can easily be checked that the following sets have the same consequences.
\begin{gather*} 
\Gamma \cup \llbracket \neg F_{i_{1}}( \vvarphi_{i_{1}}) \rrbracket  \cup \, \dots \, \cup \llbracket \neg F_{i_{n}}( \vvarphi_{i_{n}}) \rrbracket  \cup \llbracket  \neg G(\vec{\psi}, \forall x\varphi (x)) \rrbracket\; , \\ 
\Gamma \cup \llbracket  \neg F_{i_{1}}( \vvarphi_{i_{1}}) \wedge \, \dots \, \wedge \neg F_{i_{n}}( \vvarphi_{i_{n}}) \wedge \neg G(\vec{\psi}, \forall x\varphi (x)) \rrbracket\; , \\
\Gamma \cup \llbracket  \neg (F_{i_{1}}( \vvarphi_{i_{1}}) \vee \, \dots \, \vee F_{i_{n}}( \vvarphi_{i_{n}}) \vee  G(\vec{\psi}, \forall x\varphi (x))) \rrbracket\; .
\end{gather*}
Then $\Gamma \cup \llbracket  \neg (F_{i_{1}}( \vvarphi_{i_{1}}) \vee \, \dots \, \vee  F_{i_{n}}( \vvarphi_{i_{n}}) \vee  G(\vec{\psi}, \forall x\varphi (x))) \rrbracket$  is $\Cv$-consistent.

$\Gamma$ is a set of formulas from the basic language and so contains no witness variables.  Consequently there are only finitely many witness variables that occur in $\Gamma, \vvarphi_{i_{1}}, \, \dots, \,  \vvarphi_{i_{n}}, \vec{\psi}$ and $\forall x \varphi (x)$.  Let $a$ be the first witness variable that does not occur. Then, by Proposition \ref{existentialinstantiation}  , $\Gamma \cup \llbracket  \neg (F_{i_{1}}( \vvarphi_{i_{1}}) \vee \, \dots \, \vee  F_{i_{n}}( \vvarphi_{i_{k}}) \vee  G(\vec{\psi},\varphi (a))) \rrbracket$  is $\Cv$-consistent.  It follows that  $\Gamma \cup \llbracket  \neg F_{i_{1}}( \vvarphi_{i_{1}}) \rrbracket  \cup \, \dots \, \cup \llbracket  \neg F_{i_{n}}( \vvarphi_{i_{n}}) \rrbracket  \cup \llbracket  \neg G(\vec{\psi}, \varphi (a)) \rrbracket$ is $\Cv$-consistent.
\end{proof}

\begin{pro}[Basic expansion]
\label{basicexpansion}
Let $\C$ be a variant closed and axiomatically appropriate constant specification for the basic language, $\Cv$ be its extension, and let $\Gamma \subseteq \Fj$ be a $\C$-consistent set. Then there is some $\Gamma\p \subseteq \Fjv$ such that $\Gamma \subseteq \Gamma\p$, $\Gamma\p$ is a $\Cv$-maximally consistent set, and $\Gamma\p$ admits instantiation.  
\end{pro}

\begin{proof}
We define a sequence of sets of $\Fjv$ formulas $\Gamma_{1}, \Gamma_{2}, \Gamma_{3}, \dots $ so that:
\begin{itemize}
\item $\Gamma_{n}$ is $\Cv$-consistent.
\item $\Gamma_{n}$ is $\Gamma \cup \llbracket  \neg F_{i_{1}}( \vvarphi_{i_{1}}) \rrbracket  \cup \, \dots \, \cup \llbracket  \neg F_{i_{k}}( \vvarphi_{i_{k}}) \rrbracket $ (where $k\geq0$).    
\end{itemize}    

To begin, set $\Gamma_{1} =\Gamma$. By the remark at the end of Section~\ref{langextsection}, $\Gamma_{1}$ is $\Cv$-consistent because it is $\C$-consistent.  

Now, suppose $\Gamma_{n}$ has been constructed and it is of the form
$\Gamma \cup \llbracket \neg F_{i_{1}}( \vvarphi_{i_{1}}) \rrbracket  \cup \, \dots \, \cup \llbracket  \neg F_{i_{k}}( \vvarphi_{i_{k}}) \rrbracket$. Let $\bl F_{n}(\tvp_{n}), \vvarphi_{n} \br$ be the $n$\textsuperscript{th} member of the enumerated sequence. If the last formula in $\vvarphi_{n}$ is not a universal formula, let $\Gamma_{n+1} = \Gamma_{n}$. Otherwise we proceed as follows. Assume $\vvarphi_{n}$ is $\vec{\psi}, \forall x \varphi (x)$.  Say $F_{n}(\tvp_{n})$ is the disjunctive template $G(\vec{\tq},\tr)$, and so  $\llbracket \neg  F_{n}( \vvarphi_{n}) \rrbracket = \llbracket \neg G(\vec{\psi}, \forall x \varphi (x)) \rrbracket $.

If $\Gamma_{n}\cup \llbracket  \neg G(\vec{\psi}, \forall x \varphi (x)) \rrbracket$ is not $\Cv$-consistent, then take $\Gamma_{n+1}$ to be $\Gamma_{n}$.

If $\Gamma_{n}\cup \llbracket  \neg G(\vec{\psi}, \forall x \varphi (x)) \rrbracket$ is $\Cv$-consistent, then by Lemma \ref{beforeBe} there is a witness variable $a$ such that the set $\Gamma_{n} \cup \llbracket  \neg G(\vec{\psi}, \varphi (a)) \rrbracket$ is consistent. We take $\Gamma_{n+1}$ to be $\Gamma_{n} \cup \llbracket  \neg G(\vec{\psi}, \varphi (a)) \rrbracket$.

It can easily be checked that $\bigcup_{n\in \omega} \Gamma_{n}$ is $\Cv$-consistent. Then by Proposition \ref{lind} there is a set $\Gamma\p$ such that  $\bigcup_{n\in \omega} \Gamma_{n} \subseteq\Gamma\p$ and $\Gamma\p$ is $\Cv$-maximal consistent. Clearly, $\Gamma\subseteq\bigcup_{n\in \omega} \Gamma_{n} \subseteq\Gamma\p$. Now we show that $\Gamma\p$ admits instantiation.

Let $\vvarphi$ be a sequence of Henkin formulas, $\forall x \varphi (x)$ be a Henkin formula and $F(\tvp,\tq)$ be a disjunctive template. Suppose that $\Gamma\p \cup  \llbracket  \neg F(\vvarphi, \forall x\varphi (x)) \rrbracket $ is $\Cv$-consistent. For some $k \in \omega$, $\bl F(\tvp,\tq) , \bl\vvarphi,\forall x \varphi (x)  \br\br$ is the $k$\textsuperscript{th} term of the enumerated sequence. Since $\Gamma_{k}\subseteq\bigcup_{n\in \omega} \Gamma_{n} \subseteq\Gamma\p$, $\Gamma_{k} \cup  \llbracket  \neg F(\vvarphi, \forall x\varphi (x)) \rrbracket$ is $\Cv$-consistent. By construction, for some witness variable $a$, $\Gamma_{k+1}=\Gamma_{k} \cup  \llbracket  \neg F(\vvarphi, \varphi (a)) \rrbracket$ is $\Cv$-consistent. Thus $\llbracket \neg F(\vvarphi, \varphi (a)) \rrbracket \subseteq \Gamma\p$ and hence $\Gamma\p \cup \llbracket  \neg F(\vvarphi, \varphi (a)) \rrbracket $ is $\Cv$-consistent. 
\end{proof}

\begin{lemma}
\label{betweenlema}
Suppose $\Gamma$ is a set of formulas that admits instantiation, $F(\tvp)$ is a disjunctive template, and $\vvarphi$ is a sequence of Henkin formulas. Then, $\Gamma \cup \llbracket \neg F(\vvarphi) \rrbracket$ also admits instantiation.
\end{lemma}

\begin{proof}
Let $\vec{\psi}$ be a sequence of Henkin formulas, $\forall x \varphi (x)$ a Henkin formula and $G(\vec{\tq}, \tr)$ a disjunctive template. Suppose $(\Gamma \cup \llbracket \neg F(\vvarphi) \rrbracket) \cup \llbracket \neg G(\vec{\psi}, \forall x \varphi (x)) \rrbracket$ is $\Cv$-consistent.

As before, we can assume that there is no overlap between the propositional variables $\tvp, \vec{\tq}$ and  $\tr$, so $F(\tvp) \vee G(\vec{\tq}, \tr)$ is a disjunctive template. And as before, the sets
\begin{gather*} 
(\Gamma \cup \llbracket \neg F(\vvarphi) \rrbracket) \cup \llbracket \neg G(\vec{\psi}, \forall x \varphi (x)) \rrbracket\; ,\\ 
\Gamma \cup \llbracket \neg F(\vvarphi) \wedge \neg G(\vec{\psi}, \forall x \varphi (x)) \rrbracket\; ,\\
\Gamma \cup \llbracket  \neg (F(\vvarphi) \vee  G(\vec{\psi}, \forall x \varphi (x))) \rrbracket
\end{gather*}
have the same consequences. Then $\Gamma \cup \llbracket  \neg (F(\vvarphi) \vee  G(\vec{\psi}, \forall x \varphi (x))) \rrbracket$ is $\Cv$-consistent. Since $\Gamma$ admits instantiation, there is a witness variable $a$ such that $\Gamma \cup \llbracket \neg (F(\vvarphi) \vee  G(\vec{\psi}, \varphi (a))) \rrbracket$ is $\Cv$-consistent. Hence, $(\Gamma \cup \llbracket \neg F(\vvarphi) \rrbracket) \cup \llbracket \neg G(\vec{\psi},\varphi (a)) \rrbracket$ is $\Cv$-consistent.
\end{proof}

\begin{pro} [Secondary expansion]
\label{secondaryexpansion}
Let $\C$ be a variant closed and axiomatically appropriate constant specification for the basic language, $\Cv$ be its extension, and $\Gamma \subseteq \Fjv$ be a $\Cv$-consistent set that admits instantiation. Then there is some $\Gamma\p \subseteq \Fjv$ such that $\Gamma \subseteq \Gamma\p$, $\Gamma\p$ is $\Cv$-maximally consistent, and $\Gamma\p$ admits instantiation.  
\end{pro}

\begin{proof}
The proof is very similar to the proof of Proposition \ref{basicexpansion}.

We define a sequence $\Gamma_{1},\Gamma_{2}, \dots$ of $\Cv$-consistent sets that admit instantiation. To begin, $\Gamma_{1} = \Gamma$.

Next, suppose $\Gamma_{n}$ has been constructed. Let $\bl F_{n}(\tvp_{n}), \vvarphi_{n} \br$ be the $n$\textsuperscript{th} pair of the enumerated sequence. If the last term of the sequence $\vvarphi_{n}$ is not a universal formula, let $\Gamma_{n+1} = \Gamma_{n}$. Otherwise, proceed as follows. $\vvarphi_{n}$ is of the form $\vec{\psi}, \forall x \varphi (x)$. And $F_{n}(\tvp_{n})$ is the disjunctive template $G(\vec{\tq},\tr)$ and so  $\llbracket  \neg F_{n}( \vvarphi_{n}) \rrbracket  = \llbracket  \neg G(\vec{\psi}, \forall x \varphi (x)) \rrbracket$. If $\Gamma_{n}\cup \llbracket  \neg G(\vec{\psi}, \forall x \varphi (x)) \rrbracket$ is not $\Cv$-consistent, then take $\Gamma_{n+1}$ as $\Gamma_{n}$.
If $\Gamma_{n}\cup \llbracket  \neg G(\vec{\psi}, \forall x \varphi (x)) \rrbracket$ is $\Cv$-consistent then, since $\Gamma_{n}$ admits instantiation, there is a witness variable $a$ such that $\Gamma_{n}\cup \llbracket  \neg G(\vec{\psi}, \varphi (a)) \rrbracket$ is $\Cv$-consistent. By Lemma \ref{betweenlema},  $\Gamma_{n}\cup \llbracket  \neg G(\vec{\psi}, \varphi (a)) \rrbracket$ admits instantiation. So, take $\Gamma_{n+1}$ as $\Gamma_{n}\cup \llbracket  \neg G(\vec{\psi}, \varphi (a)) \rrbracket$.

As before, it can be checked that $\bigcup_{n\in \omega}\Gamma _{n}$ is a  $\Cv$-consistent set that admits instantiation. By Proposition \ref{lind} there is a set $\Gamma\p$ such that  $\bigcup_{n\in \omega} \Gamma_{n} \subseteq\Gamma\p$ and $\Gamma\p$ is $\Cv$-maximal consistent. It is easy to see that  $\Gamma\p$  admits instantiation. 
\end{proof}

\section{The Canonical Model}\label{canonicalmodelsection}

We now define the canonical model and show it behaves as expected.  The domain will, of course, be the set of witness variables.

\begin{defn}[Canonical model]
\label{canonical}
A \emph{canonical model} $\M = \model$, using constant specification $\C$, is defined as follows.
\begin{itemize}
\item $\W$ is the set of all $\C(\textbf{V})$-maximally consistent sets that admit instantiation.
\item Let $\Gamma, \Delta \in \W$. $\Gamma\R\Delta$ iff $\Gamma^{\#} \subseteq  \Delta$ (as defined in \ref{sharp}).
\item $\D = \textbf{V}$.
\item For an $n$-place relation symbol $P$ and for $\Gamma \in \W$, let $\I(P,\Gamma)$ be the set of all $\vec{a}$ where $\vec{a} \in \textbf{V}$ and $P(\vec{a}) \in \Gamma$.
\item For $\Gamma \in \W$, set $\Gamma \in \E(t,\varphi)$ iff $\lc{t}_{X}\varphi \in \Gamma$, where $\lc{t}_{X}\varphi$ is a closed Henkin formula and $X$ is the set of witness variables in $\varphi$.
\end{itemize}
\end{defn}

First we need to check that $\M$ is indeed a Fitting model meeting $\C$. Since the argument is similar to the one presented in \cite[pp. 13-14]{Fitting14} we are only going to consider a few of the cases.
\smallskip

\textbf{$\R$ is reflexive}. Let $\Gamma \in \W$ and suppose $\forall \vec{y}\varphi(\vec{y})\in\Gamma^\#$.  Then there is a closed Henkin formula $\lc{t}_{X}\varphi(\vec{y})\in\Gamma$ such that $\vec{y}$ is an $n$-ary sequence of basic variables none of which occur in $X$. By repeated use of axiom \textbf{B5} and classical reasoning, $\turn_{\C(\textbf{V})}\lc{t}_{X}\varphi(\vec{y}) \rightarrow \lc{\gen_{\vec{y}}}_{X}\forall \vec{y} \varphi(\vec{y})$.  By axiom \textbf{B1}, $\turn_{\C(\textbf{V})}\lc{\gen_{\vec{y}}}_{X}\forall \vec{y} \varphi(\vec{y}) \rightarrow \forall \vec{y} \varphi(\vec{y})$
hence, by the maximal consistency of $\Gamma$, $\forall \vec{y} \varphi(\vec{y}) \in \Gamma$. Thus $\Gamma^{\#} \subseteq \Gamma$, i.e., $\Gamma\R\Gamma$.
\smallskip

\textbf{$\R$ is transitive}. Let $\Gamma, \Delta, \Theta \in \W$ such that $\Gamma\R\Delta$ and $\Delta\R\Theta$.  Assume $\forall \vec{y} \psi(\vec{a},\vec{y}) \in \Gamma^{\#}$, where $\vec{a}$ is a sequence of witness variables and $\vec{y}$ is a sequence of basic variables.  Then for some justification term $t$, $ \lc{t}_{\{\vec{a}\}}\psi(\vec{a},\vec{y}) \in \Gamma$.

By axiom \textbf{B4} and by the maximal consistency of $\Gamma$,  $\lc{!t}_{\{\vec{a}\}} \lc{t}_{\{\vec{a}\}}\psi(\vec{a},\vec{y}) \in \Gamma$. Since $\lc{t}_{\{\vec{a}\}}\psi(\vec{a},\vec{y})$ has no free basic variables and $\Gamma \R \Delta$, then $\lc{t}_{\{\vec{a}\}}\psi(\vec{a},\vec{y}) \in \Delta$. And since $\Delta\R\Theta$, then $\forall \vec{y} \psi(\vec{a},\vec{y}) \in \Theta$. Thus, $\Gamma^{\#} \subseteq \Theta$, i.e., $\Gamma\R\Theta$.
\smallskip

\textbf{$!$ Condition}. Let $\Gamma \in \W$. Suppose that $\Gamma \in \E(t, \varphi (\vec{a}))$ and $X$ is a set of witness variables such that $\{\vec{a}\} \subseteq X$. Since $\Gamma \in \E(t, \varphi (\vec{a}))$, $\lc{t}_{\{\vec{a}\}}\varphi (\vec{a}) \in \Gamma$. Using axiom \textbf{A3} and the maximality of $\Gamma$, $\lc{t}_X\varphi (\vec{a}) \in \Gamma$. Then by axiom \textbf{B4}, $\lc{!t}_X\lc{t}_X\varphi (\vec{a}) \in \Gamma$, i.e., $\Gamma \in \E(!t, \lc{t}_X\varphi (\vec{a}))$. 
\smallskip

\textbf{$\jb$ Condition}. The proof is by contradiction.  Suppose $\Gamma \in \W$ and, for every $d \in \textbf{V}$, $\Gamma \in \E(t, \varphi (\vec{a}, d))$ but $\Gamma \notin \E(\jb(t), \forall y\varphi (\vec{a}, y))$.  Equivalently using the definition of $\E$ in the canonical model, $\lc{t}_{\{\vec{a}, d\}} \varphi (\vec{a}, d) \in \Gamma$ for every $d \in \textbf{V}$, but $\lc{\jb(t)}_{\{\vec{a}\}} \forall y \varphi (\vec{a}, y) \notin \Gamma$.  We show this leads to contradiction.

By the maximality of $\Gamma$, $\neg \lc{\jb(t)}_{\{\vec{a}\}} \forall y \varphi (\vec{a}, y) \in \Gamma$.  By axiom \textbf{Bb} and classical reasoning, $\turn_{\C(\textbf{V})} \neg \lc{\jb(t)}_{\{\vec{a}\}} \forall y \varphi(\vec{a},y) \rightarrow \neg \forall y \lc{t}_{\{\vec{a},y\}}\varphi(\vec{a},y)$ and thus, $\neg \forall y \lc{t}_{\{\vec{a},y\}}\varphi(\vec{a},y) \in \Gamma$. By Proposition~\ref{instantiation} there is some $a^{*} \in \textbf{V}$ such that $\neg  \lc{t}_{\{\vec{a},a^{*}\}}\varphi(\vec{a},a^{*}) \in \Gamma$. But this contradicts the fact that $\lc{t}_{\{\vec{a}, d\}} \varphi (\vec{a}, d) \in \Gamma$ for every $d \in \textbf{V}$.
\smallskip

\textbf{$\C$ Condition}.
Recall that a model $\M$ meets constant specification $\C$ provided, whenever $\lc{c}\varphi \in \C$, then $\E(c,\varphi) = \W$. Let $\lc{c}\varphi$ be a member of $\C$ and $\Gamma$ a member of $\W$. Since $\Gamma$ is a $\Cv$-maximally consistent set, $\C \subseteq \Cv \subseteq \Gamma$. Hence, $\lc{c}\varphi \in \Gamma$ and from this it follows that $\Gamma \in \E(c,\varphi)$.
\smallskip

We have shown that the canonical model is a Fitting model meeting $\C$. Now, we have a version of the usual Truth Lemma.

\begin{lemma}[Truth lemma]
\label{truthlemma}
Let $\M=\model$ be a canonical model. For each $\Gamma \in \W$ and for each closed Henkin formula $\varphi$,
\[
\M,\Gamma \Vdash  \varphi \; \text{iff} \;\varphi \in \Gamma\; .
\]
\end{lemma}

\begin{proof}
By induction on the degree of $\varphi$. The crucial cases are when $\varphi$ is $\lc{t}_{X}\psi$ and when $\varphi$ is $\forall x \psi(x)$, and these are the only ones we discuss.  Assume the induction hypothesis holds for formulas less complex than $\varphi$.
\smallskip

Assume that $\varphi$ is $\lc{t}_{X}\psi$.

($\Rightarrow$) Suppose $\lc{t}_{X}\psi \notin \Gamma$. Let $X\p \subseteq X$ be the set containing exactly the witness variables that occur in $\psi$. It is not the case that $\lc{t}_{X\p}\psi\in \Gamma$, because otherwise, by axiom \textbf{A3} and by the maximal consistency of $\Gamma$,  $\lc{t}_{X}\psi \in \Gamma$. So by the definition of $\E$, $\Gamma \notin \E(t,\psi)$, and thus $\M,\Gamma \nmodels \lc{t}_{X}\psi$. 

($\Leftarrow$) Suppose $\lc{t}_{X}\psi \in \Gamma$. Let $X\p \subseteq X$ be as above. Then by the axiom \textbf{A2} and the maximal consistency of $\Gamma$, $\lc{t}_{X\p}\psi \in \Gamma$, and hence $\Gamma \in \E(t,\psi)$. Now, let $\Delta \in \W$ such that $\Gamma \R \Delta$. Since $\Gamma^{\#} \subseteq\Delta$, we have that  $\forall \vec{y}\psi \in \Delta$ where $\vec{y}$ are the free basic variables of $\psi$. Then by classical logic and the maximal consistency of $\Delta$, for every $\vec{a} \in \textbf{V}$,  $\psi(\vec{a}) \in \Delta$. By the induction hypothesis, for every $\vec{a} \in \textbf{V}$, $\M, \Delta \Vdash  \psi(\vec{a})$. Therefore, $\M,\Gamma \Vdash  \lc{t}_{X\p}\psi$, and so $M,\Gamma \Vdash  \lc{t}_{X}\psi$.
\medskip

Assume that $\varphi$ is $\forall x \psi(x)$.

($\Rightarrow$) Suppose $\forall x \psi(x) \notin \Gamma$. By the maximal consistency of $\Gamma$, $\neg \forall x \psi(x) \in \Gamma$. Since $\Gamma$ admits instantiation, then by Proposition \ref{instantiation} there is an $a \in \textbf{V}$ such that $\neg \psi(a) \in \Gamma$. By the consistency of $\Gamma$, $\psi(a) \notin \Gamma$. By the induction hypothesis, $\M, \Gamma \nmodels \psi(a)$, thus $\M, \Gamma \nmodels \forall x \psi(x)$.    

($\Leftarrow$) Suppose  $\forall x \psi(x) \in \Gamma$. By the classical axioms and the maximal consistency of $\Gamma$, for every $a \in \textbf{V}$,  $\psi(a) \in \Gamma$. By the induction hypothesis, $\M, \Gamma \Vdash  \psi(a)$, for every $a \in \textbf{V}$. Therefore,  $\M, \Gamma \Vdash  \forall x \psi(x)$.  
\end{proof}

\begin{theor}[Completeness for \folpb]
\label{compleLP}
Let $\C$ be a variant closed and axiomatically appropriate constant specification. For every closed formula $\varphi \in \Fj$, if $\Vdash _{\C} \varphi$, then $\turn_{\C}\varphi$.
\end{theor}

\begin{proof}
Suppose $\not\turn_{\C}\varphi$. Then $\{\neg \varphi\}$ is $\C$-consistent. By Proposition \ref{basicexpansion}, there is a $\C(\textbf{V})$-maximally consistent $\Gamma$ such that $\Gamma$ admits instantiation and  $\{\neg \varphi\} \subseteq \Gamma$. By the Truth lemma, $\M,\Gamma \Vdash  \neg \varphi$, so  $\M,\Gamma \nmodels \varphi$. Hence, $\nmodels_{\C} \varphi$.     
\end{proof}

\begin{defn}[Fully explanatory] 
A model $\M = \model$ is \emph{fully explanatory} if the following condition is fulfilled. Let $\varphi$ be a formula with no free individual variables, but with constants from the domain of the model. Let $w \in \W$. If for every $v \in \W$ such that $w\R v$, $\M, v \Vdash  \varphi$, then there is a justification term $t \in \Tj$ such that $\M, w \Vdash   \lc{t}_{X}\varphi$, where $X$ is the set of domain constants appearing in $\varphi$. 
\end{defn}

\begin{theor}
Let $\C$ be a variant closed and axiomatically appropriate constant specification. The canonical model of \folpb\ using $\C$ is fully explanatory.
\end{theor}

\begin{proof}
Let $\M = \model$ be a canonical model, $\Gamma \in \W$, $\varphi$ a closed Henkin formula and $X$ the set of the witness variables occurring $\varphi$. We shall show that if $\M, \Gamma \nmodels \lc{t}_{X}\varphi$ for every justification term $t \in \Tj$, then there is a $\Delta \in \W$ such that $\Gamma \R\Delta$ and $\M, \Delta \nmodels \varphi$.

If $\M, \Gamma \nmodels \lc{t}_{X}\varphi$ for every justification term $t \in \Tj$ then by the Truth Lemma, $\neg \lc{t}_{X}\varphi \in \Gamma$ for every justification term $t \in \Tj$. The template $G(\tp) =\tp$ is a disjunctive template. Let $F(\tp) = \Box G(\tp)$. Then $\llbracket \neg F(\varphi)\rrbracket \subseteq \Gamma$, and so, trivially, $\Gamma \cup \llbracket \neg F(\vvarphi)\rrbracket$ is $\Cv$-consistent. By item 2.\ of Proposition \ref{up}, $\Gamma^{\#} \cup \llbracket \neg G(\varphi)\rrbracket$ is $\Cv$-consistent, i.e.,  $\Gamma^{\#} \cup \{  \neg \varphi \} $ is $\Cv$-consistent. By Proposition \ref{instantiationsharp}, $\Gamma^{\#}$ admits instantiation, and then by Lemma \ref{betweenlema},  $\Gamma^{\#} \cup \{  \neg \varphi \} $ admits instantiation. By Proposition \ref{secondaryexpansion}, there is a $\Cv$-maximal consistent set $\Delta$ such that $\Delta$ admits instantiation and $\Gamma^{\#} \cup \{  \neg \varphi \}\subseteq\Delta$. Since $\Gamma^{\#} \subseteq \Delta$, $\Gamma \R \Delta$. And since $\neg \varphi \in \Delta$, by the Truth Lemma,  $\M, \Delta \nmodels \varphi$.
\end{proof}

\section{The Case Of First-Order \jtff} 

The definitions and results for first-order \jtff\ (\fojtff) follow with little change from the ones related to \folpb. There are only a few minor details  that need to be addressed.

On the syntactic side there are two differences.  First, propositional \jtff\ expands the language of \jlp\ by adding the unary justification term $?$, and this also happens in the first-order case; an axiom \textbf{B6} for ? is added. Second, $\jb$ is not a primitive justification term of \fojtff, and consequently the axiom \textbf{Bb} is dropped.   Here is the axiom that is added.   
\begin{center}
\begin{tabular}{c c}    
\textbf{B6} & $\neg \lc{t}_{X}\varphi \rightarrow \lc{?t}_{X} \neg\lc{t}_{X}\varphi\; .$    
\end{tabular}
\end{center}

In this section, the provability predicates $\turn$ and $\turn_{\C}$ will refer to \fojtff. It is easy to verify that the Deduction Lemma, the Internalization Theorem and Theorem \ref{cbarcanburidan} hold for \fojtff.

By an argument due to Prior, one can derive the Barcan Formula in first-order \msfi. The situation in \fojtff\ is analogous. Fix an axiomatically appropriate constant specification. Then for every term $t$ a justification term that we may call $\jb(t)$ can be constructed, playing the role that our primitive term $\jb(t)$ played for \folpb.  That is, the \textbf{Bb} formula, $\forall y \lc{t}_{Xy}\varphi(y) \rightarrow \lc{\jb(t)}_{X} \forall y \varphi(y)$, is provable in \fojtff\ for this defined term $\jb(t)$.

\begin{pro}[Explicit counterpart of the Barcan Formula]
Let $\C$ be an axiomatically appropriate constant specification and $y$ be an individual variable. For every finite set of individual variables $X$ such that $y \notin X$, for every formula $\varphi(y)$ and every justification term $t$, there is a justification term $\jb(t)$ such that $\turn_{\C} \forall y \lc{t}_{Xy}\varphi(y) \rightarrow \lc{\jb(t)}_{X} \forall y \varphi(y)$.
\end{pro}
\begin{proof}\ \\
\begin{center}  
\begin{longtable}{ll}
1.  $ \forall y \lc{t}_{Xy}\varphi(y) \rightarrow \lc{t}_{Xy}\varphi(y)$ &  \emph{classical axiom.} \\[0.3cm]
2.  $\lc{c_{1}}(\forall y \lc{t}_{Xy}\varphi(y) \rightarrow \lc{t}_{Xy}\varphi(y))$  &  \emph{constant specification.} \\[0.3cm]
3. $\lc{c_{1}}_{Xy}(\forall y \lc{t}_{Xy}\varphi(y) \rightarrow \lc{t}_{Xy}\varphi(y))$ &  \emph{from 2 by \textbf{A3}.}  \\[0.3cm]
4. $ (\forall y \lc{t}_{Xy}\varphi(y) \rightarrow \lc{t}_{Xy}\varphi(y)) \rightarrow (\neg \lc{t}_{Xy}\varphi(y) \rightarrow \neg \forall y \lc{t}_{Xy}\varphi(y) )$ &  \emph{classical axiom.} \\[0.3cm]
5.  $\lc{c_{2}}((\forall y \lc{t}_{Xy}\varphi(y) \rightarrow \lc{t}_{Xy}\varphi(y)) \rightarrow (\neg \lc{t}_{Xy}\varphi(y) \rightarrow \neg \forall y \lc{t}_{Xy}\varphi(y) ))$ &  \emph{constant specification.} \\[0.3cm]        
6.  $\lc{c_{2}}_{Xy}((\forall y \lc{t}_{Xy}\varphi(y) \rightarrow \lc{t}_{Xy}\varphi(y)) \rightarrow (\neg \lc{t}_{Xy}\varphi(y) \rightarrow \neg \forall y \lc{t}_{Xy}\varphi(y) ))$ &  \emph{from 5 by \textbf{A3}.} \\[0.3cm]
7.  $\lc{[c_{2}\cdot c_{1}]}_{Xy}(\neg \lc{t}_{Xy}\varphi(y) \rightarrow \neg \forall y \lc{t}_{Xy}\varphi(y) )$  &  \emph{from 6 and 3 by \textbf{B2}.} \\[0.3cm]
8.  $\lc{?t}_{Xy}\neg \lc{t}_{Xy}\varphi(y) \rightarrow \lc{[[c_{2}\cdot c_{1}] \cdot ?t]}_{Xy}\neg \forall y \lc{t}_{Xy}\varphi(y) $  &  \emph{from 7 by \textbf{B2}.} \\[0.3cm]   
9.  $\neg\lc{[[c_{2}\cdot c_{1}] \cdot ?t]}_{Xy}\neg \forall y \lc{t}_{Xy}\varphi(y)  \rightarrow \neg\lc{?t}_{Xy}\neg \lc{t}_{Xy}\varphi(y)$  &  \emph{from 8 by classical reasoning.} \\[0.3cm]           
10. $\neg \lc{?t}_{Xy}\neg \lc{t}_{Xy}\varphi(y) \rightarrow \varphi(y) $ &  \emph{\jtff\ theorem.} \\[0.3cm]               
11. $\neg\lc{[[c_{2}\cdot c_{1}] \cdot ?t]}_{Xy}\neg \forall y \lc{t}_{Xy}\varphi(y)  \rightarrow \varphi(y) $ &  \emph{from 9 and 10.} \\[0.3cm]       
12. $\lc{[[c_{2}\cdot c_{1}] \cdot ?t]}_{Xy}\neg \forall y \lc{t}_{Xy}\varphi(y)  \rightarrow \lc{[[c_{2}\cdot c_{1}] \cdot ?t]}_{X}\neg \forall y \lc{t}_{Xy}\varphi(y)  $ &  \emph{\textbf{A2}.} \\[0.3cm]
13. $\neg \lc{[[c_{2}\cdot c_{1}] \cdot ?t]}_{X}\neg \forall y \lc{t}_{Xy}\varphi(y)  \rightarrow \neg\lc{[[c_{2}\cdot c_{1}] \cdot ?t]}_{Xy}\neg \forall y \lc{t}_{Xy}\varphi(y)  $ &  \emph{from 12 by classical reasoning.} \\[0.3cm]
14. $\neg\lc{[[c_{2}\cdot c_{1}] \cdot ?t]}_{X}\neg \forall y \lc{t}_{Xy}\varphi(y)  \rightarrow \varphi(y) $ &  \emph{from 11 and 13.} \\[0.3cm]               
15. $\forall y(\neg\lc{[[c_{2}\cdot c_{1}] \cdot ?t]}_{X}\neg \forall y \lc{t}_{Xy}\varphi(y)  \rightarrow \varphi(y) )$ &  \emph{generalization.} \\[0.3cm]                
16. $\neg\lc{[[c_{2}\cdot c_{1}] \cdot ?t]}_{X}\neg \forall y \lc{t}_{Xy}\varphi(y)   \rightarrow \forall y\varphi(y) $ &  \emph{$y \notin X$ and classical reasoning.} \\[0.3cm]
17. $\lc{r}(\neg\lc{[[c_{2}\cdot c_{1}] \cdot ?t]}_{X}\neg \forall y \lc{t}_{Xy}\varphi(y)   \rightarrow \forall y\varphi(y) )$ &  \emph{internalization.} \\[0.3cm]
18. $\lc{r}_{X}(\neg\lc{[[c_{2}\cdot c_{1}] \cdot ?t]}_{X}\neg \forall y \lc{t}_{Xy}\varphi(y)   \rightarrow \forall y\varphi(y) )$ &  \emph{from 17 by \textbf{A3}.} \\[0.3cm]     
19. $\lc{?[[c_{2}\cdot c_{1}] \cdot ?t]}_{X}\neg\lc{[[c_{2}\cdot c_{1}] \cdot ?t]}_{X}\neg \forall y \lc{t}_{Xy}\varphi(y)   \rightarrow \lc{[r \cdot ?[[c_{2}\cdot c_{1}] \cdot ?t]]}_{X}\forall y\varphi(y) $ &  \emph{from 18 by \textbf{B2}.} \\[0.3cm]
20.  $\forall y \lc{t}_{Xy}\varphi(y) \rightarrow\lc{?[[c_{2}\cdot c_{1}] \cdot ?t]}_{X}\neg\lc{[[c_{2}\cdot c_{1}] \cdot ?t]}_{X}\neg \forall y \lc{t}_{Xy}\varphi(y)$ &  \emph{\jtff\ theorem.}  \\[0.3cm]            
21.  $\forall y \lc{t}_{Xy}\varphi(y) \rightarrow \lc{[r \cdot ?[[c_{2}\cdot c_{1}] \cdot ?t]]}_{X}\forall y\varphi(y)$ &  \emph{from 19 and 20.}  \\[0.3cm]                                            
\end{longtable}
\end{center}
\end{proof}

Before defining the models for \fojtff\ we impose one more condition on the evidence function $\E$.

\begin{defn}[Strong evidence] 
Let  $\M = \model$ be a Fitting model. We say that $\E$ is a \emph{strong evidence function} if for every term $t$ and Henkin formula $\varphi$,  $\E(t,\varphi) \subseteq \{w \in \W \; | \; \M,w \Vdash \lc{t}_{X}\varphi\}$; where $X$ is the set of constants occurring in $\varphi$.
\end{defn}  

\begin{defn}
A \emph{Fitting model for \fojtff} is a Fitting model $\M = \model$ where $\bl \W, \R, \D \br$ is an \fojtff-skeleton (as defined in \ref{defskeleton}); $\E$ is a strong evidence function; and the following holds:
\begin{itemize}
\item[] \textbf{$?$ Condition} $\W  \backslash \E (t, \varphi) \subseteq \E(?t,\neg \lc{t}_{X} \varphi)$, where $X$ is the set of constants occurring in $\varphi$. 
\end{itemize}       
\end{defn}  

\begin{theor}[Soundness]
Let $\C$ be a constant specification. For every formula $\varphi \in \Fj$, if $\turn_{\C} \varphi$, then $\Vdash _{\C}\varphi$.
\end{theor}    

\begin{proof}   
We will check only the new axiom of \fojtff. Suppose $\varphi$ is an instance of \textbf{B6}, i.e., $\varphi$ is $\neg \lc{t}_{X}\psi \rightarrow \lc{?t}_{X}\neg\lc{t}_{X}\psi$. For notational simplicity, assume $X= \{x\}$ and $\psi = \psi(x,y)$. Then we have that $\turn_{\C}\neg\lc{t}_{\{x\}}\psi(x,y) \rightarrow \lc{?t}_{\{x\}}\neg\lc{t}_{\{x\}}\psi(x,y)$.    

Let $\M = \model$ be a Fitting model for \fojtff\ meeting $\C$, $w \in \W$ and $a \in \D$. Suppose $\M, w \Vdash  \neg \lc{t}_{\{a\}}\psi(a,y)$. Then, $\M, w \nmodels \lc{t}_{\{a\}}\psi(a,y)$. By the definition of the strong evidence function, $w \notin \E (t, \psi(a,y))$. By the ? condition, $w \in \E(?t,\neg \lc{t}_{\{a\}}\psi(a,y))$. Again, by the strong evidence function $\M, w \Vdash  \lc{?t}_{\{a\}}\neg \lc{t}_{\{a\}}\psi(a,y)$.  
\end{proof}     

The argument for the Completeness Theorem is developed essentially as before. Let a canonical model $\M = \model$ be defined essentially as in Definition~\ref{canonical}, except that now the underlying logic is \fojtff. There are only a few new features of $\M$ that we need to show:  $\R$ should be an equivalence relation and $\E$ should be a strong evidence function that satisfies the $?$ Condition. 
\smallskip

\textbf{$\R$ is symmetric}. Let $\Gamma, \Delta \in \W$. Suppose that $\Gamma \R \Delta$ but it is not the case that $\Delta \R \Gamma$.  We derive a contradiction.

By our assumptions, $\Gamma^\#\subseteq\Delta$ but $\Delta^{\#} \nsubseteq \Gamma$.  From the latter, for some term $t$, set of witness variables $X$ and Henkin formula $\varphi(\vec{y})$,  $\lc{t}_{X}\varphi(\vec{y}) \in \Delta$ but $\forall \vec{y} \varphi(\vec{y}) \notin \Gamma$.

Suppose we had that $\lc{t}_{X}\varphi(\vec{y}) \in \Gamma$. Then by repeated use of axiom \textbf{B5}, $ \lc{\gen_{\vec{y}}}_{X}\forall \vec{y} \varphi(\vec{y})\in \Gamma$, and by axiom \textbf{B1}, $\forall \vec{y} \varphi(\vec{y})\in \Gamma$, a contradiction. Hence, $\lc{t}_{X}\varphi(\vec{y}) \notin \Gamma$.

By what we have just shown, and the maximal consistency of $\Gamma$, $\neg \lc{t}_{X}\varphi(\vec{y}) \in \Gamma$. By axiom \textbf{B6}, $\lc{?t}_{X}\neg \lc{t}_{X}\varphi(\vec{y}) \in \Gamma$. Since $\Gamma^{\#} \subseteq \Delta$, then $\neg \lc{t}_{X}\varphi(\vec{y}) \in \Delta$, and we have a contradiction. Therefore, if $\Gamma \R \Delta$, then $\Delta \R \Gamma$.
\smallskip

\textbf{$?$ Condition}. Suppose $\Gamma \in \W \backslash \E(t,\varphi)$.  Then by definition of $\E$, $\lc{t}_{X}\varphi \notin \Gamma$, where $X$ is the set of all witness variables occurring in $\varphi$.  By the maximal consistency of $\Gamma$,  $\neg \lc{t}_{X}\varphi \in \Gamma$, and by the axiom \textbf{B6}, $\lc{?t}_{X}\neg \lc{t}_{X}\varphi \in \Gamma$. Hence, $\Gamma \in \E(?t,\neg \lc{t}_{X}\varphi)$. 
\smallskip

\textbf{Strong Evidence}. Using the Truth Lemma, we have the following chain of implications:
\[
\Gamma \in \E(t,\varphi) \Rightarrow \lc{t}_{X}\varphi \in \Gamma \Rightarrow \M,\Gamma \Vdash\lc{t}_{X}\varphi \Rightarrow \Gamma \in \{\Delta \in \W \; |\; \M,\Delta \Vdash \lc{t}_{X}\varphi\}\;.
\]
Thus, $\E$ is an strong evidence function.
\smallskip

And so the theorems of the previous section can be carried over to \fojtff.

\begin{theor}[Completeness for \fojtff]
Let $\C$ be a variant closed and axiomatically appropriate constant specification. For every closed formula $\varphi \in \Fj$, if $\Vdash _{\C} \varphi$, then $\turn_{\C}\varphi$.
\end{theor}

\begin{theor}
Let $\C$ be a variant closed and axiomatically appropriate constant specification. The canonical model of \fojtff\ using $\C$ is fully explanatory.
\end{theor}

\section{Future Work}
There is both a constructive and a non-constructive proof of the Realization Theorem for first-order \jlp. The constructive one uses a cut-free sequent calculus \cite{Artemov11}, and the non-constructive one uses the Model Existence Theorem \cite{Fitting13}. How to adapt these methods successfully to \folpb\ and \fojtff\ will be the subject of further exploration. The Completeness Theorem presented here should be a useful tool in this endeavor.

One more remark is worth mentioning. In \cite{Fine79} Kit Fine showed that the Interpolation Theorem fails for constant domain first-order modal logic. The reason for this failure is understood to be a lack of expressiveness of the quantified modal logic. With hybrid logic we have an example of an extension of modal logic that can be used to repair this theorem \cite{Areces03}, and thus it is natural to ask \emph{if the novelties introduced by justification logic are expressive enough to restore the Interpolation Theorem}.

Craig's Interpolation Theorem could be formulated for \folpb, perhaps as follows. 
\begin{itemize}
\item[] \emph{The Interpolation Theorem} holds for \folpb\ iff for every constant specification $\C$ and sentences $\varphi$ and $\psi$, if $\turn_{\C} \varphi \rightarrow \psi$, then there is a formula $\theta$ such that $\turn_{\C} \varphi \rightarrow \theta$, $\turn_{\C} \theta \rightarrow \psi$ and the non-logical symbols and the justification terms that occur in $\theta$ occur both in $\varphi$ and $\psi$.
\end{itemize}

Among the systems studied by Fine is constant domain first-order \msf\ (\fomsf). As the reader can easily verify, \emph{if the Realization Theorem holds between \folpb\ and \fomsf, then it follows that the Interpolation Theorem fails for \folpb.} 
Investigation of the Realization Theorem for \folpb\ can also reveal something of the expressive power of justification logic; hence this topic is not only a subject of interest to the researchers involved in justification logic, but is also a topic of interest for the broader modal logic community.

\end{document}